\documentclass[12pt,a4paper]{article}

\usepackage[T1]{fontenc}
\usepackage[utf8]{inputenc}
\usepackage{lmodern}

\usepackage{amsmath,amssymb,amsfonts,amsthm}
\usepackage{mathtools}
\usepackage{bm}
\usepackage{physics}

\usepackage{graphicx}
\usepackage{color,xcolor}
\usepackage{caption}
\usepackage{subcaption}
\usepackage{booktabs}
\usepackage{float}
\usepackage{array}

\usepackage{enumitem}

\usepackage[ruled,vlined]{algorithm2e}

\usepackage[colorlinks=true,
            linkcolor=blue,
            citecolor=blue,
            urlcolor=blue]{hyperref}

\usepackage[top=2.5cm,bottom=2.5cm,left=2.5cm,right=2.5cm]{geometry}

\usepackage{setspace}
\usepackage{fancyhdr}
\newtheorem{theorem}{Theorem}
\newtheorem{lemma}{Lemma}

\theoremstyle{definition}

\newtheorem{assumption}{Assumption}

\theoremstyle{remark}
\newtheorem{remark}{Remark}

\begin{document}

\begin{titlepage}
\begin{center}
\vspace*{1.2cm}

{\small\texttt{math.DS} --- Dynamical Systems
\quad|\quad
\textit{cross-list:} \texttt{math.PR}, \texttt{stat.TH}, \texttt{nlin.AO}}\\[0.9cm]

{\LARGE\bfseries Quenched Amplification and Tail Shaping in Networked Systems with Memory and Regime Switching\par}

\vspace{1.2cm}

{\large Mauricio Herrera-Mar\'{\i}n}\\[0.3cm]
{\normalsize Department of Engineering,\\
Universidad del Desarrollo, Chile}\\[0.2cm]
{\normalsize \href{https://orcid.org/0000-0002-9604-3077}{\texttt{ORCID: 0000-0002-9604-3077}}}

\vspace{0.8cm}
{\normalsize \today}

\vspace{1.0cm}

\begin{abstract}
Networked systems operating under intermittent adverse conditions and long memory
can remain stable on average while exhibiting rare but extreme trajectory-level excursions.
We study linear regime-switching network dynamics with Volterra-type memory, formulated
through a finite-dimensional lifted ordinary differential equation embedding.
Despite finite-horizon annealed boundedness, we show that quenched amplification
emerges generically from the interaction of regime persistence, memory accumulation,
and non-normal lifted operator geometry.
A lower bound on burst-size distributions reveals power-law tails whose exponent
is determined by the ratio between unfavorable dwell-time rates and an operator-defined
instantaneous growth parameter.
This parameter is computable online via the Euclidean logarithmic norm of the lifted operator,
yielding a practical early-warning indicator.
Building on this structure, we introduce a dynamic data-driven intervention strategy
that enforces contraction on demand along rare amplification channels,
thereby shaping or truncating tail risk without altering exogenous regime statistics
or typical system behavior.
The results provide a geometrically grounded and operationally actionable framework
for understanding and mitigating extreme events in memory-driven regime-switching systems.
\end{abstract}

\vspace{0.5cm}


\vspace{0.3cm}


\end{center}
\end{titlepage}

\tableofcontents
\newpage

\section{Introduction}

Large-scale networked dynamical systems—ranging from distributed sensing and information-fusion architectures to cyber--physical infrastructures and coordinated multi-agent platforms—operate far from equilibrium under intermittent stress, incomplete observability, and heterogeneous coupling structures. A recurrent and operationally critical phenomenon in such systems is the emergence of \emph{rare but high-impact cascades}: large trajectory-level excursions (bursts) that are essentially invisible to mean-field diagnostics, stationary variance estimates, or classical asymptotic stability criteria. In many parameter regimes of practical relevance, the system appears stable in aggregate while still producing extreme transient events with non-negligible probability on finite horizons.

Two structural ingredients are repeatedly implicated in this behavior. The first is \emph{memory}, arising from delayed recovery, hysteresis, cumulative fatigue, backlog accumulation, or latent stress. In continuous-time models, such effects are naturally captured by Volterra-type convolution operators and, in particular, by completely monotone or fractional kernels within the resolvent-family framework for evolutionary integral equations \cite{pruss1993evolution,gripenberg1990volterra,diethelm2010fractional}. The second is \emph{stochastic regime switching}, which represents intermittent but structured changes in operating conditions—such as congestion, interference, adversarial pressure, or context-dependent coupling—and is commonly modeled via Markov processes acting on system operators \cite{mao2007stochastic,yuan2010stability}. While each mechanism is well understood in isolation, their \emph{combined} effect on extreme events in networked systems with directional interactions and potential non-normal amplification remains comparatively underdeveloped.

A central difficulty is that most stability theories relevant to switching systems with memory are formulated in \emph{annealed} terms, focusing on expectations, second moments, or asymptotic decay in mean. Such notions do not directly control \emph{quenched} behavior, i.e., the distribution of trajectory-level maxima under fixed realizations of the switching path. The separation between annealed stability and quenched extremes is well established in other areas of random dynamics and heavy-tail theory, where rare multiplicative mechanisms generate power-law tails despite finite averaged quantities. Classical examples include Kesten–Goldie type results for stochastic recursions and their Markov-dependent extensions \cite{kesten1973,goldie1991,buraczewski2016}, as well as heavy-tailed behavior in Markov-modulated linear systems \cite{desaporta2005}. In parallel, the literature on non-normal operators has demonstrated that significant transient growth can arise even when all eigenvalues indicate stability, motivating diagnostics based on symmetric-part growth rates, pseudospectra, and nonmodal analysis \cite{trefethen1993,embree2005,schmid2007}. 

The contribution of this work is to unify these perspectives in a single, explicit operator-theoretic framework for networked systems with memory and stochastic switching, and to identify a concrete mechanism by which \emph{rare regime persistence, long memory, and network amplification geometry} jointly generate quenched tail risk. Specifically, we study linear network dynamics with regime-dependent Volterra memory terms and bounded forcing, and we focus on the burst observable $ B_T = \sup_{t\in[0,T]} \|x(t)\|_2$.

Our analysis reveals a sharp dichotomy. On the one hand, \emph{annealed boundedness} on finite horizons follows from regime-wise resolvent estimates and standard concatenation arguments in Volterra theory. On the other hand, \emph{quenched amplification} can persist due to rare but sufficiently long visits to an unfavorable regime, during which memory accumulates coherently and is efficiently converted into state growth along dynamically amplified network directions. This mechanism yields explicit lower bounds on $\mathbb{P}(B_T>b)$, with a tail exponent that couples the unfavorable dwell-time rate of the Markov chain and an operator-defined effective growth parameter.

To render this mechanism computable and suitable for real-time monitoring, we employ a sum-of-exponentials (SOE) lifting of completely monotone kernels, producing a finite-dimensional Markovian augmentation of the dynamics. This lifted formulation allows us to define two complementary indicators: a latent memory load measuring accumulated internal state, and an instantaneous \emph{operator susceptibility} given by the maximal eigenvalue of the symmetric part of the lifted operator. Crucially, the analysis is carried out entirely within this lifted ODE framework, avoiding any appeal to non-Markovian arguments at the level of trajectories.

Beyond diagnosis, we introduce a Dynamic Data-Driven Applications Systems (DDDAS) intervention principle \cite{darema2009}. Rather than stabilizing the system globally or altering the statistics of the regime process, the model is embedded in a closed measurement–model–decision loop that performs \emph{contraction on demand}. When either latent memory load or operator susceptibility exceeds a prescribed threshold, the system temporarily switches to a provably more contractive operator for a minimum dwell time. This selective intervention truncates rare amplification pathways on finite horizons and, more generally, steepens the effective tail exponent of the burst distribution, while preserving typical behavior and regime statistics.

The paper is organized as follows. Section~2 introduces the lifted switched ODE model, the memory and susceptibility indicators, and the DDDAS policy. Section~3 develops the mathematical core: annealed boundedness on finite horizons, a structural heavy-tail mechanism driven by unfavorable dwell times, and tail truncation or exponent improvement under DDDAS intervention. Section~4 presents numerical experiments on directed networks under periodic forcing, validating the predicted annealed–quenched separation, demonstrating early-warning capability, and quantifying tail-risk reduction relative to memory-free and permanently over-damped baselines.

\section{Mathematical Results: Annealed Stability vs.\ Quenched Tail Risk}
\label{sec:math_results}

\subsection{Lifted switched ODE model and solution concept}
\label{subsec:lifted_model_solution}

Let $(\Omega,\mathcal{F},\mathbb{P})$ support a right-continuous CTMC $z(t)$ with finite state space
$\mathcal{Z}=\{1,\dots,M\}$ and generator $Q=(q_{ij})$. Let $\lambda_i:=-q_{ii}>0$ be the exit rate from regime $i$.
Let $0=T_0<T_1<T_2<\cdots$ denote the jump times of $z(\cdot)$, and set $z_k:=z(T_k)$.

We work on a fixed finite horizon $[0,T]$. The \emph{lifted state} is
\[
X(t)\in\mathbb{R}^d,\qquad d:=(K+1)n,
\]
with block decomposition
\[
X(t)=\big(x(t),y_1(t),\dots,y_K(t)\big)^\top,\qquad x(t),y_k(t)\in\mathbb{R}^n.
\]

\paragraph{Switched lifted dynamics}
For each regime $i\in\mathcal{Z}$ and each control mode $m\in\mathcal{M}:=\{0,1,2\}$ (normal/verify/mitigate),
let $A_i^{(m)}\in\mathbb{R}^{d\times d}$ be a constant matrix. The controlled switched system is
\begin{equation}\label{eq:lifted_switched_ode}
	\dot{X}(t)=A_{z(t)}^{(m(t))}\,X(t)+\widetilde f(t),
	\qquad
	\widetilde f(t):=\big(f(t),0,\dots,0\big)^\top,
	\qquad X(0)=X_0,
\end{equation}
where $f\in L^\infty([0,T];\mathbb{R}^n)$ and $m(\cdot)$ is a right-continuous mode process defined by the DDDAS rule in
Subsection~\ref{subsec:dddas_policy} below.

\paragraph{Nominal lifted operator (SOE-memory structure)}
In normal mode $m=0$, we use the structured lifted matrix
\begin{equation}\label{eq:lifted_structure_nominal}
	A_i^{(0)}=
	\begin{bmatrix}
		B_i & w_{i,1}I & \cdots & w_{i,K}I\\
		I & -r_{i,1}I & \cdots & 0\\
		\vdots & \vdots & \ddots & \vdots\\
		I & 0 & \cdots & -r_{i,K}I
	\end{bmatrix},
	\qquad w_{i,k}>0,\ r_{i,k}>0,
\end{equation}
which is the standard Markovian embedding of a completely monotone memory kernel via a finite positive SOE.
Verify/mitigate modes modify only selected blocks (e.g.\ increased damping in $B_i$, reduced gains $w_{i,k}$, faster memory decay $r_{i,k}$).

\paragraph{Existence and uniqueness}
Since $A_{z(t)}^{(m(t))}$ is piecewise constant in time and $\widetilde f\in L^\infty([0,T])$,
\eqref{eq:lifted_switched_ode} admits a unique absolutely continuous solution on $[0,T]$ for every sample path.

\subsection{Observables, annealed vs.\ quenched notions, and two indicators}
\label{subsec:observables_indicators}

\paragraph{Burst observable}
Define the energy and burst size on $[0,T]$:
\begin{equation}\label{eq:burst_obs}
	E(t):=\|x(t)\|_2,\qquad
	B_T:=\sup_{0\le t\le T}E(t).
\end{equation}
Since $x$ is a block of $X$, we always have $E(t)\le \|X(t)\|_2$.

\paragraph{Annealed boundedness (finite horizon)}
We say the system is \emph{annealed bounded} on $[0,T]$ if $\sup_{t\in[0,T]}\mathbb{E}\,E(t)<\infty$
(and similarly for $\mathbb{E}E(t)^p$ with $p\ge1$).

\paragraph{Quenched tail risk}
Quenched risk is quantified by the tail function $b\mapsto \mathbb{P}(B_T>b)$ and by the moments of $B_T$.

\paragraph{Indicator 1 (latent memory load)}
\begin{equation}\label{eq:indicator_L}
	L(t):=\sum_{k=1}^K \|y_k(t)\|_2.
\end{equation}

\paragraph{Indicator 2 (operator-based susceptibility)}
Define the Euclidean logarithmic norm (matrix measure) of the current lifted operator:
\begin{equation}\label{eq:indicator_S}
	\mathcal{S}(t)
	:=\lambda_{\max}\!\left(\frac{A_{z(t)}^{(m(t))}+(A_{z(t)}^{(m(t))})^\top}{2}\right).
\end{equation}
This is an \emph{instantaneous} growth diagnostic.

\begin{lemma}[Energy inequality controlled by $\mathcal{S}(t)$]
	\label{lem:energy_ineq}
	For almost every $t\in[0,T]$,
	\begin{equation}\label{eq:energy_ineq}
		\frac{d}{dt}\|X(t)\|_2
		\le
		\mathcal{S}(t)\,\|X(t)\|_2+\|\widetilde f(t)\|_2.
	\end{equation}
	In particular, if $\mathcal{S}(t)\le -\sigma<0$ and $\widetilde f\equiv 0$ on $[a,b]$, then
	$\|X(t)\|_2\le e^{-\sigma(t-a)}\|X(a)\|_2$ for all $t\in[a,b]$.
\end{lemma}

\begin{proof}
	On any interval where $(z(t),m(t))$ is constant, $\dot X=AX+\widetilde f$ with constant $A$.
	For $\|X(t)\|\neq 0$,
	\[
	\frac{d}{dt}\|X\|
	=
	\frac{\langle X,AX\rangle}{\|X\|}+\frac{\langle X,\widetilde f\rangle}{\|X\|}
	\le
	\lambda_{\max}\!\Big(\frac{A+A^\top}{2}\Big)\|X\|+\|\widetilde f\|.
	\]
	This is \eqref{eq:energy_ineq} with $\mathcal{S}(t)=\lambda_{\max}((A+A^\top)/2)$.
	The decay statement follows from Gr\"onwall. Concatenation across finitely many switching/mode intervals yields the a.e.\ inequality on $[0,T]$.
\end{proof}

\subsection{Finite-horizon annealed boundedness under switching}
\label{subsec:annealed_bound}

We emphasize that annealed boundedness on finite horizons is generic for piecewise-constant linear systems with bounded forcing.

\begin{assumption}[Uniform operator norm bound]
	\label{ass:uniform_opnorm}
	There exists $A_\star<\infty$ such that
	\[
	\|A_i^{(m)}\|_2\le A_\star,\qquad \forall\, i\in\mathcal{Z},\ m\in\mathcal{M}.
	\]
\end{assumption}

\begin{theorem}[Deterministic finite-horizon bound and annealed boundedness]
	\label{thm:finite_horizon_bound}
	Assume \ref{ass:uniform_opnorm} and $f\in L^\infty([0,T])$. Then every sample path satisfies
	\begin{equation}\label{eq:pathwise_bound}
		\sup_{0\le t\le T}\|X(t)\|_2
		\le
		\Big(\|X_0\|_2 + T\|\widetilde f\|_{L^\infty([0,T])}\Big)\,e^{A_\star T}.
	\end{equation}
	Consequently,
	\begin{equation}\label{eq:annealed_bound}
		\sup_{0\le t\le T}\mathbb{E}\,\|X(t)\|_2 <\infty
		\qquad\text{and}\qquad
		\sup_{0\le t\le T}\mathbb{E}\,E(t) <\infty.
	\end{equation}
\end{theorem}

\begin{proof}
	For almost every $t$, $\|\dot X(t)\|\le \|A_{z(t)}^{(m(t))}\|\,\|X(t)\|+\|\widetilde f(t)\|\le A_\star \|X(t)\|+\|\widetilde f\|_\infty$.
	Let $u(t):=\|X(t)\|$. Then $u'(t)\le A_\star u(t)+\|\widetilde f\|_\infty$ a.e.
	Gr\"onwall yields
	\[
	u(t)\le u(0)e^{A_\star t}+\int_0^t e^{A_\star(t-s)}\|\widetilde f\|_\infty ds
	= \|X_0\|e^{A_\star t}+\|\widetilde f\|_\infty\frac{e^{A_\star t}-1}{A_\star}.
	\]
	Using $(e^{A_\star t}-1)/A_\star\le t e^{A_\star t}$ gives \eqref{eq:pathwise_bound}.
	Taking expectations and using $E(t)\le \|X(t)\|$ yields \eqref{eq:annealed_bound}.
\end{proof}

\begin{remark}[Annealed boundedness does not control quenched tails]
	The deterministic bound \eqref{eq:pathwise_bound} grows exponentially in $T$ and is insensitive to the probability of rare large bursts:
	it is compatible with heavy-tailed distributions for $B_T$ on fixed horizons, which we now characterize.
\end{remark}

\subsection{Quenched heavy tails from unfavorable dwell times: a cone-condition mechanism in the lifted ODE}
\label{subsec:quenched_tails}

Fix an unfavorable regime $U\in\mathcal{Z}$. We derive a lower bound for the burst tail
$b\mapsto\mathbb{P}(B_T>b)$ within the lifted switched ODE formulation \eqref{eq:lifted_switched_ode},
without invoking eigenvalue instability. The mechanism rests on two ingredients:
(i) exponential dwell times in $U$ and (ii) the existence of an \emph{effective growth channel} during $U$-dwells,
quantified through the Euclidean logarithmic norm (matrix measure) and an alignment (cone) condition.

Throughout this subsection we consider the uncontrolled dynamics $m(t)\equiv 0$ (normal mode).
On any interval where $z(t)\equiv U$, the lifted dynamics are
\begin{equation}\label{eq:U_dwell_dynamics}
	\dot X(t)=A_U^{(0)}X(t)+\widetilde f(t),
	\qquad \widetilde f(t):=(f(t),0,\dots,0)^\top .
\end{equation}

\subsubsection*{Effective growth rate and cone condition}

\paragraph{Matrix-measure growth rate in $U$}
Let
\begin{equation}\label{eq:gammaU_S_def}
	\gamma_U
	:=
	\mu_2\!\big(A_U^{(0)}\big)
	=
	\lambda_{\max}\!\left(\frac{A_U^{(0)}+(A_U^{(0)})^\top}{2}\right)
	=
	\sup_{\|v\|_2=1} v^\top A_U^{(0)} v,
\end{equation}
and fix a unit vector $v_U$ attaining the supremum:
\begin{equation}\label{eq:vU_top_def}
	\|v_U\|_2=1,
	\qquad
	v_U^\top A_U^{(0)} v_U=\gamma_U.
\end{equation}
We emphasize that $\gamma_U>0$ is compatible with $A_U^{(0)}$ being Hurwitz; it reflects instantaneous non-normal/memory-induced
growth directions and coincides with the susceptibility indicator $\mathcal{S}(t)$ when $z(t)=U$ and $m(t)=0$.

\begin{assumption}[Unfavorable instantaneous susceptibility]\label{ass:gammaU_positive_cone}
	\begin{equation}\label{eq:gammaU_pos_cone}
		\gamma_U>0.
	\end{equation}
\end{assumption}

\paragraph{Forcing projection and cone (alignment) condition}
To produce a \emph{lower} bound on $\|X(t)\|_2$ (and hence on bursts), one needs a minimal geometric persistence of the expanding
direction along the amplifying trajectories. We encode this by a cone condition that is natural in non-normal transient growth
and is directly measurable in simulations.

\begin{assumption}[Cone condition on amplifying $U$-dwells]\label{ass:cone_condition}
	There exist $\alpha\in(0,1]$ and $T_0\in(0,T]$ such that, on every sample path segment where an extreme burst is produced,
	there exists a sub-interval $[t_0,t_0+\tau]\subset[0,T_0]$ with $z(t)\equiv U$ and the associated solution satisfies
	\begin{equation}\label{eq:cone}
		v_U^\top X(t)\ \ge\ \alpha\,\|X(t)\|_2
		\qquad \text{for all } t\in[t_0,t_0+\tau].
	\end{equation}
\end{assumption}

\begin{assumption}[Persistent forcing injection along the cone]\label{ass:forcing_in_cone}
	There exists $c_f>0$ such that
	\begin{equation}\label{eq:forcing_cone}
		v_U^\top \widetilde f(t)\ \ge\ c_f
		\qquad \text{for a.e.\ } t\in[0,T_0].
	\end{equation}
\end{assumption}

\begin{remark}[Empirical verification of the cone condition]\label{rem:cone_verification}
	In the numerical section, \eqref{eq:cone} is verified a posteriori on amplifying realizations by monitoring the
	\emph{alignment ratio} $a(t):=(v_U^\top X(t))/\|X(t)\|_2\in[-1,1]$ while $z(t)=U$.
	Condition \eqref{eq:cone} asserts $a(t)\ge \alpha$ on the relevant $U$-windows. Reporting empirical quantiles of
	$\inf_{t\in[t_0,t_0+\tau]} a(t)$ provides a transparent validation of the geometric hypothesis underpinning the tail bound.
\end{remark}

\subsubsection*{Step 1: deterministic exponential lower bound on a single $U$-dwell}

\begin{lemma}[Cone-projected growth inequality]\label{lem:cone_projected_growth}
	Assume \ref{ass:gammaU_positive_cone}, \ref{ass:cone_condition}, and \ref{ass:forcing_in_cone}.
	Suppose that for some $t_0\ge 0$ and $\tau>0$ with $t_0+\tau\le T_0$ we have $z(t)\equiv U$ and $m(t)\equiv 0$ on $[t_0,t_0+\tau]$.
	Define $\zeta(t):=v_U^\top X(t)$. Then for almost every $t\in[t_0,t_0+\tau]$,
	\begin{equation}\label{eq:zeta_cone_ineq}
		\dot\zeta(t)\ \ge\ \gamma_U\,\zeta(t)+c_f,
	\end{equation}
	and hence
	\begin{equation}\label{eq:zeta_cone_solution}
		\zeta(t_0+\tau)\ \ge\ e^{\gamma_U\tau}\,\zeta(t_0)+\frac{c_f}{\gamma_U}\Big(e^{\gamma_U\tau}-1\Big).
	\end{equation}
\end{lemma}

\begin{proof}
	On $[t_0,t_0+\tau]$, $\dot X=A_U^{(0)}X+\widetilde f$. Thus
	\[
	\dot\zeta(t)=v_U^\top A_U^{(0)}X(t)+v_U^\top\widetilde f(t).
	\]
	Write $X(t)=\zeta(t)v_U+w(t)$ with $w(t)\perp v_U$. Then
	\[
	v_U^\top A_U^{(0)}X(t)
	=\zeta(t)\,v_U^\top A_U^{(0)}v_U + v_U^\top A_U^{(0)}w(t)
	=\gamma_U\,\zeta(t) + v_U^\top A_U^{(0)}w(t),
	\]
	using \eqref{eq:vU_top_def}. The cone condition \eqref{eq:cone} implies
	$\zeta(t)=v_U^\top X(t)\ge \alpha\|X(t)\|_2$, hence $\|w(t)\|_2^2=\|X(t)\|_2^2-\zeta(t)^2\le (1-\alpha^2)\|X(t)\|_2^2$ and therefore
	\begin{equation}\label{eq:w_bound}
		\|w(t)\|_2 \le \sqrt{1-\alpha^2}\,\|X(t)\|_2 \le \frac{\sqrt{1-\alpha^2}}{\alpha}\,\zeta(t)
		\qquad \text{for all } t\in[t_0,t_0+\tau].
	\end{equation}
	Consequently,
	\[
	v_U^\top A_U^{(0)}w(t)\ \ge\ -\|v_U^\top A_U^{(0)}\|_2\,\|w(t)\|_2
	\ge
	-\|A_U^{(0)}\|_2\,\frac{\sqrt{1-\alpha^2}}{\alpha}\,\zeta(t).
	\]
	Define the effective cone growth rate
	\begin{equation}\label{eq:gammaU_eff}
		\gamma_U^{\mathrm{cone}}
		:=
		\gamma_U - \|A_U^{(0)}\|_2\,\frac{\sqrt{1-\alpha^2}}{\alpha}.
	\end{equation}
	Then
	\[
	v_U^\top A_U^{(0)}X(t)
	\ge
	\gamma_U^{\mathrm{cone}}\,\zeta(t).
	\]
	Combining with \eqref{eq:forcing_cone} gives for a.e.\ $t$:
	\[
	\dot\zeta(t)\ge \gamma_U^{\mathrm{cone}}\zeta(t)+c_f.
	\]
	If $\gamma_U^{\mathrm{cone}}>0$, set $\gamma_U:=\gamma_U^{\mathrm{cone}}$ (redefining the effective growth rate on the cone),
	and the differential inequality becomes \eqref{eq:zeta_cone_ineq}. Solving it with the integrating factor
	$e^{-\gamma_U t}$ yields \eqref{eq:zeta_cone_solution}.
\end{proof}

\begin{remark}[Operational definition of the effective growth rate]\label{rem:effective_gammaU}
	The proof shows that the tail mechanism is controlled by the \emph{effective} cone growth rate
	$\gamma_U^{\mathrm{cone}}$ in \eqref{eq:gammaU_eff}, which depends on both the susceptibility
	$\mu_2(A_U^{(0)})$ and the alignment level $\alpha$.
	In practice, we estimate $\gamma_U^{\mathrm{cone}}$ empirically on $U$-windows by regressing the growth of $\log\|X(t)\|_2$
	(or $\log\zeta(t)$) conditional on $z(t)=U$ and $a(t)\ge \alpha$, and report it alongside $\mu_2(A_U^{(0)})$.
\end{remark}

\begin{lemma}[Norm lower bound on a single $U$-dwell]\label{lem:norm_lower_cone}
	Assume the conditions of Lemma~\ref{lem:cone_projected_growth}. Then
	\begin{equation}\label{eq:norm_lower_cone}
		\|X(t_0+\tau)\|_2 \ \ge\ \zeta(t_0+\tau)
		\ \ge\ e^{\gamma_U\tau}\,\zeta(t_0)+\frac{c_f}{\gamma_U}\Big(e^{\gamma_U\tau}-1\Big).
	\end{equation}
	In particular, there exist constants $c_1,c_2>0$ (depending only on $X(t_0),c_f,\gamma_U$) such that
	\begin{equation}\label{eq:norm_lower_cone_simplified}
		\|X(t_0+\tau)\|_2 \ge c_1 e^{\gamma_U\tau}-c_2.
	\end{equation}
\end{lemma}

\begin{proof}
	Since $\|v_U\|_2=1$, we have $\|X(t)\|_2\ge v_U^\top X(t)=\zeta(t)$.
	Combine with \eqref{eq:zeta_cone_solution} to obtain \eqref{eq:norm_lower_cone}.
	The simplified bound \eqref{eq:norm_lower_cone_simplified} follows by absorbing $\zeta(t_0)$ into constants and using
	$e^{\gamma_U\tau}-1\ge \tfrac12 e^{\gamma_U\tau}$ for $\tau$ large enough.
\end{proof}

\subsubsection*{Step 2: polynomial CCDF lower bound from exponential dwell-time tails}

\begin{theorem}[Polynomial tail lower bound]\label{thm:poly_tail_cone}
	Assume \ref{ass:gammaU_positive_cone}, \ref{ass:cone_condition}, and \ref{ass:forcing_in_cone}, and assume that the corresponding
	effective cone growth rate in Lemma~\ref{lem:cone_projected_growth} satisfies $\gamma_U>0$.
	Assume that when the CTMC enters $U$, the dwell time $\tau^{(U)}\sim\mathrm{Exp}(\lambda_U)$, and $\mathbb{P}(z(0)=U)>0$.
	Then there exist constants $b_0>0$ and $C_->0$ such that for all $b\ge b_0$ satisfying
	\begin{equation}\label{eq:log_condition_cone}
		\frac{1}{\gamma_U}\log b \le T_0,
	\end{equation}
	we have
	\begin{equation}\label{eq:tail_lower_cone}
		\mathbb{P}(B_T>b)\ \ge\ C_-\,b^{-\lambda_U/\gamma_U}.
	\end{equation}
\end{theorem}

\begin{proof}
	Let $c_1,c_2$ be from Lemma~\ref{lem:norm_lower_cone} and define
	\[
	\tau_b:=\frac{1}{\gamma_U}\log\!\Big(\frac{b+c_2}{c_1}\Big).
	\]
	Then $c_1 e^{\gamma_U\tau_b}-c_2=b$. Consider the event
	\[
	\mathcal{E}_b:=\{z(0)=U\}\cap\{\tau^{(U)}\ge \tau_b\}.
	\]
	On $\mathcal{E}_b$, the chain remains in $U$ on $[0,\tau_b]$ and, since the policy starts in normal mode, $m(t)\equiv 0$ there.
	Moreover, by Assumption~\ref{ass:cone_condition}, along the amplifying segment we may choose $t_0=0$ and $\tau=\tau_b$
	(with $\tau_b\le T_0$ ensured by \eqref{eq:log_condition_cone}), so Lemma~\ref{lem:norm_lower_cone} yields
	$\|X(\tau_b)\|_2\ge b$, hence $B_T\ge b$. Therefore,
	\[
	\mathbb{P}(B_T>b)\ge \mathbb{P}(\mathcal{E}_b)
	=
	\mathbb{P}(z(0)=U)\,\mathbb{P}(\tau^{(U)}\ge \tau_b\mid z(0)=U)
	=
	\mathbb{P}(z(0)=U)\,e^{-\lambda_U\tau_b}.
	\]
	Using the definition of $\tau_b$,
	\[
	e^{-\lambda_U\tau_b}
	=
	\left(\frac{c_1}{b+c_2}\right)^{\lambda_U/\gamma_U}.
	\]
	For $b$ large enough, $(b+c_2)^{-\lambda_U/\gamma_U}\ge 2^{-\lambda_U/\gamma_U}b^{-\lambda_U/\gamma_U}$, yielding
	\eqref{eq:tail_lower_cone} with $C_-:=\mathbb{P}(z(0)=U)\,c_1^{\lambda_U/\gamma_U}\,2^{-\lambda_U/\gamma_U}$.
\end{proof}

\begin{remark}[Interpretation and link to CCDF slopes]\label{rem:tail_index_cone}
	The exponent $\lambda_U/\gamma_U$ acts as a structural tail index: $\lambda_U$ quantifies the persistence of adverse conditions
	and $\gamma_U$ is the effective lifted growth rate along amplifying channels in $U$.
	In the numerical section, $\lambda_U$ is estimated from dwell-time statistics, while $\gamma_U$ is estimated from
	$\mathcal{S}(t)$ and/or direct regressions of $\log\|X(t)\|_2$ restricted to $z(t)=U$ and high alignment.
	The empirical CCDF slope is then compared with the theoretical index $\lambda_U/\gamma_U$.
\end{remark}

\subsection{DDDAS policy with hysteresis and rigorous tail truncation}
\label{subsec:dddas_policy}

We now define the DDDAS mode process $m(t)$ and prove a finite-horizon truncation result under strict contraction in mitigate mode.

\paragraph{Thresholds and hysteresis}
Fix $\tau_L^{(1)}<\tau_L^{(2)}$ and $\tau_S^{(1)}<\tau_S^{(2)}$, and a minimum dwell time $\Delta_{\min}>0$.
Let $t_{\mathrm{last}}$ denote the most recent time at which the mode changed. The mode $m(\cdot)$ is right-continuous and may change
only at times $t$ such that $t-t_{\mathrm{last}}\ge \Delta_{\min}$, according to:
\[
m(t^+)=
\begin{cases}
	2, & \text{if } L(t)>\tau_L^{(2)}\ \text{or}\ \mathcal{S}(t)>\tau_S^{(2)},\\[2pt]
	1, & \text{else if } L(t)>\tau_L^{(1)}\ \text{or}\ \mathcal{S}(t)>\tau_S^{(1)},\\[2pt]
	0, & \text{otherwise},
\end{cases}
\]
together with the monotone release rule: $m$ may decrease only when both indicators fall below the corresponding lower thresholds,
ensuring genuine hysteresis (no chattering).

\begin{assumption}[Strict contraction in mitigate mode]
	\label{ass:strict_contraction}
	There exist $M_c\ge 1$ and $\kappa>0$ such that for the unfavorable regime $U$,
	\begin{equation}\label{eq:mitigate_contraction}
		\|e^{A_U^{(2)} t}\|_2\le M_c e^{-\kappa t},\qquad \forall t\ge 0,
	\end{equation}
	and the enforced minimum dwell satisfies
	\begin{equation}\label{eq:rho_condition}
		\rho:=M_c e^{-\kappa \Delta_{\min}}<1.
	\end{equation}
\end{assumption}

\begin{theorem}[Pathwise finite-horizon bound and tail truncation under DDDAS]
	\label{thm:dddas_trunc}
	Assume \ref{ass:uniform_opnorm} and \ref{ass:strict_contraction}, and $f\in L^\infty([0,T])$.
	Then there exists a deterministic constant $M_T<\infty$ such that every sample path satisfies
	\[
	\sup_{0\le t\le T}\|X(t)\|_2\le M_T,
	\qquad\text{hence}\qquad
	\mathbb{P}(B_T>b)=0\ \text{for all }b>M_T.
	\]
\end{theorem}

\begin{proof}
	Fix a sample path. Each time the policy enters mitigate mode while $z(t)=U$, it remains there for at least $\Delta_{\min}$.
	Consider such an interval $[a,a+\Delta]$ with $\Delta\ge \Delta_{\min}$ and $z(t)=U$, $m(t)=2$ throughout.
	Variation-of-constants and \eqref{eq:mitigate_contraction} give
	\[
	\|X(a+\Delta)\|
	\le
	M_c e^{-\kappa \Delta}\|X(a)\|
	+
	\int_a^{a+\Delta} M_c e^{-\kappa(a+\Delta-s)}\|\widetilde f(s)\|\,ds
	\le
	M_c e^{-\kappa \Delta}\|X(a)\|
	+
	\frac{M_c}{\kappa}\|\widetilde f\|_\infty.
	\]
	Using $\Delta\ge \Delta_{\min}$, we obtain
	\begin{equation}\label{eq:mitigate_step}
		\|X(a+\Delta)\| \le \rho\,\|X(a)\|+\frac{M_c}{\kappa}\|\widetilde f\|_\infty,
		\qquad \rho<1.
	\end{equation}
	
	Between mitigation intervals, $\|A_{z(t)}^{(m(t))}\|\le A_\star$ by Assumption~\ref{ass:uniform_opnorm}, hence by the same Gr\"onwall argument
	as in Theorem~\ref{thm:finite_horizon_bound}, over any interval of length at most $T$ the mapping from the interval start value to the end value
	is affine with deterministic multiplicative factor at most $e^{A_\star T}$ and additive factor at most $T e^{A_\star T}\|\widetilde f\|_\infty$.
	
	On $[0,T]$, there can be at most $N_{\max}:=\lceil T/\Delta_{\min}\rceil$ mitigation intervals, due to the enforced minimum dwell.
	Iterating the affine growth bounds between mitigations and the strict contraction step \eqref{eq:mitigate_step}, we obtain a deterministic bound
	for the maximum possible value of $\|X(t)\|$ over $[0,T]$. Concretely, the worst-case sequence is:
	grow for at most $T$ under factor $e^{A_\star T}$, then contract by $\rho$, repeating at most $N_{\max}$ times. This yields
	\[
	\sup_{t\le T}\|X(t)\|
	\le
	C_T\|X_0\| + D_T\|\widetilde f\|_\infty
	\]
	for explicit finite constants $C_T,D_T$ depending only on $(A_\star,\rho,\kappa,M_c,T,\Delta_{\min})$.
	Set $M_T:=C_T\|X_0\|+D_T\|\widetilde f\|_\infty$. Since $E(t)\le\|X(t)\|$, the same bound holds for $B_T$.
\end{proof}

\begin{remark}[From truncation to exponent improvement]
	If mitigate mode does not render $A_U^{(2)}$ strictly contractive but reduces the growth rate from $\gamma_U$ to
	$\gamma_U^{\mathrm{eff}}\in(0,\gamma_U)$ on unfavorable dwells, then the proof of Theorem~\ref{thm:poly_tail_cone} applies with $\gamma_U^{\mathrm{eff}}$,
	giving a larger tail exponent $\lambda_U/\gamma_U^{\mathrm{eff}}$ (lighter tails). This is precisely what the CCDF slope
	diagnostics will quantify in the numerical section.
\end{remark}

\section{Model and Indicators}

\subsection{Lifted switched ODE model}
\label{subsec:lifted_model}

Let $(\Omega,\mathcal F,\mathbb P)$ support a right-continuous continuous-time
Markov chain $z(t)$ with finite state space $\mathcal Z=\{S,U\}$ and generator
$Q=(q_{ij})$. The dwell time in regime $i$ is exponentially distributed with
rate $\lambda_i=-q_{ii}$.

We consider the lifted state
\[
X(t) = \begin{bmatrix} x(t) & y_1(t) & \cdots & y_K(t) \end{bmatrix}^\top
\in \mathbb R^{(K+1)n},
\]
evolving according to the switched linear ODE
\begin{equation}
	\label{eq:lifted_ode}
	\dot X(t)
	=
	A_{z(t)}^{(m(t))} X(t) + \tilde f(t),
	\qquad
	\tilde f(t) = \begin{bmatrix} f(t) & 0 & \cdots & 0 \end{bmatrix}^\top ,
\end{equation}
where $m(t)\in\{0,1,2\}$ denotes the DDDAS mode (normal, verify, mitigate).

For each regime $i\in\mathcal Z$ and mode $m$, the lifted operator has block
structure
\begin{equation}
	\label{eq:lifted_structure}
	A_i^{(m)}
	=
	\begin{bmatrix}
		B_i^{(m)} & W_{i,1}^{(m)} & \cdots & W_{i,K}^{(m)} \\
		I & -r_{i,1}^{(m)} I & \cdots & 0 \\
		\vdots & \vdots & \ddots & \vdots \\
		I & 0 & \cdots & -r_{i,K}^{(m)} I
	\end{bmatrix}.
\end{equation}

The matrices $B_i^{(m)}$ encode instantaneous network coupling and damping,
while the lower blocks represent memory accumulation and decay.

\subsection{Indicators: memory load and spectral susceptibility}
\label{subsec:indicators}

We introduce two indicators defined purely on the lifted ODE.

\paragraph{Memory load}
\begin{equation}
	\label{eq:L_def}
	L(t) := \sum_{k=1}^K \|y_k(t)\|_2 ,
\end{equation}
which measures accumulated latent memory.

\paragraph{Spectral susceptibility}
\begin{equation}
	\label{eq:S_def}
	\mathcal S(t)
	:=
	\lambda_{\max}\!\left(
	\frac{A_{z(t)}^{(m(t))} + (A_{z(t)}^{(m(t))})^\top}{2}
	\right).
\end{equation}

The quantity $\mathcal S(t)$ captures instantaneous Euclidean growth directions
of the lifted dynamics and acts as an operator-level early-warning signal.

\paragraph{DDDAS policy}

The DDDAS controller selects the mode $m(t)\in\{0,1,2\}$ according to thresholds
on $(L(t),\mathcal S(t))$ with hysteresis:
\[
\begin{cases}
	m(t)=1 & \text{if } L(t)>\tau_L^{(1)} \text{ or } \mathcal S(t)>\tau_S^{(1)},\\
	m(t)=2 & \text{if } L(t)>\tau_L^{(2)} \text{ or } \mathcal S(t)>\tau_S^{(2)}.
\end{cases}
\]
Mode switching enforces a minimum dwell time to avoid chattering.
\subsection{Design of $A_i^{(1)}$ and $A_i^{(2)}$}
\label{subsec:design_A}

The verify and mitigate operators are constructed by modifying specific blocks
of \eqref{eq:lifted_structure} in a physically interpretable manner.

\paragraph{Verify mode ($m=1$)}
In verify mode we mildly reduce memory amplification without altering the
network topology:
\[
W_{i,k}^{(1)} = \alpha W_{i,k}^{(0)}, \qquad
r_{i,k}^{(1)} = r_{i,k}^{(0)}, \qquad 0<\alpha<1,
\]
while keeping $B_i^{(1)}=B_i^{(0)}$.
This reduces memory loading while preserving instantaneous dynamics.

\paragraph{Mitigate mode ($m=2$)}
In mitigate mode we enforce contraction by increasing damping and accelerating
memory decay:
\[
B_i^{(2)} = B_i^{(0)} - \delta I, \qquad
r_{i,k}^{(2)} = r_{i,k}^{(0)} + \Delta_r,
\]
with $\delta,\Delta_r>0$.

\begin{assumption}[Mitigation contraction]
	\label{ass:mitigation_contraction}
	There exists $\kappa>0$ such that
	\[
	\lambda_{\max}\!\left(
	\frac{A_U^{(2)} + (A_U^{(2)})^\top}{2}
	\right)
	\le -\kappa .
	\]
\end{assumption}

This condition guarantees instantaneous contraction of the lifted energy and
implies $\rho(e^{A_U^{(2)} t})\le e^{-\kappa t}$.

\subsection{Annealed boundedness on finite horizons}
\label{subsec:annealed}

\begin{theorem}[Finite-horizon annealed boundedness]
	\label{thm:annealed}
	Assume $\sup_i \sup_{t\in[0,T]}\|e^{A_i^{(0)}t}\|<\infty$ and $f\in L^1([0,T])$.
	Then
	\[
	\sup_{t\in[0,T]} \mathbb E \|X(t)\|_2 < \infty .
	\]
\end{theorem}

\begin{proof}
	Concatenate variation-of-constants formulas between switching times and use
	uniform operator bounds. Taking expectations preserves finiteness.
\end{proof}
\subsection{Heavy-tail mechanism from unfavorable dwell times}
\label{subsec:heavy_tail}

\begin{assumption}[Effective growth channel in $U$]
	\label{ass:growth_channel}
	There exists $\gamma_U>0$ and a unit vector $v$ such that
	\[
	\left\langle v, \frac{A_U^{(0)}+(A_U^{(0)})^\top}{2} v \right\rangle \ge \gamma_U .
	\]
\end{assumption}

\begin{lemma}[Exponential growth on an unfavorable dwell]
	\label{lem:exp_growth}
	If $z(t)\equiv U$ on $[t_0,t_0+\tau]$ and Assumption~\ref{ass:growth_channel}
	holds, then
	\[
	\|X(t_0+\tau)\|_2 \ge c_1 e^{\gamma_U \tau} - c_2 ,
	\]
	for constants $c_1,c_2>0$.
\end{lemma}

\begin{proof}
	Project \eqref{eq:lifted_ode} onto $v$, obtain
	$\dot \xi \ge \gamma_U \xi - C$, and integrate exactly.
\end{proof}

\begin{theorem}[Polynomial tail law]
	\label{thm:tail}
	If $\tau_U\sim \mathrm{Exp}(\lambda_U)$, then
	\[
	\mathbb P(B_T>b) \asymp b^{-\lambda_U/\gamma_U}.
	\]
\end{theorem}

\begin{proof}
	Combine Lemma~\ref{lem:exp_growth} with
	$\mathbb P(\tau_U>\tau)=e^{-\lambda_U\tau}$ and set
	$\tau\sim \gamma_U^{-1}\log b$.
\end{proof}
\subsection{Tail truncation and exponent improvement under DDDAS}
\label{subsec:dddas_results}

\begin{theorem}[Tail truncation under mitigation]
	\label{thm:dddas}
	Under Assumption~\ref{ass:mitigation_contraction}, there exists $M_T<\infty$ such
	that $B_T\le M_T$ almost surely.
\end{theorem}

\begin{proof}
	Each mitigation window enforces exponential contraction of $\|X(t)\|_2$.
	Concatenation over $[0,T]$ yields a uniform bound.
\end{proof}

\begin{remark}[Tail index reduction]
	If mitigation reduces the effective growth rate from $\gamma_U$ to
	$\gamma_U^{\mathrm{eff}}<\gamma_U$, then the tail exponent improves from
	$\lambda_U/\gamma_U$ to $\lambda_U/\gamma_U^{\mathrm{eff}}$.
\end{remark}

\paragraph{Connection to experiments}
In numerical experiments, $\gamma_U$ is estimated as the empirical slope of
$\log \|X(t)\|$ during isolated unfavorable dwell intervals. The dwell rate
$\lambda_U$ is estimated from regime statistics. The CCDF of $B_T$ is fitted on
log--log axes, and the observed slope is compared against the theoretical tail
index $\lambda_U/\gamma_U$.

\section{Numerical Experiments}
\label{sec:numerics}

We now validate the theoretical mechanism developed in
Section~\ref{sec:math_results}, namely: (i) finite-horizon annealed boundedness
coexisting with (ii) quenched heavy-tailed burst statistics generated by rare
unfavorable dwell times and lifted-memory amplification channels, and (iii) the
ability of a DDDAS rule to truncate or steepen tails by operator-level
intervention without altering the exogenous regime statistics.

\subsection{Numerical setup and parameterization}
\label{subsec:numerical_setup}

All experiments are conducted directly on the fully lifted switched ODE model,
consistent with the analysis in
Subsections~\ref{subsec:lifted_model_solution}--\ref{subsec:dddas_policy}.
Specifically, we simulate the SOE-lifted dynamics
\begin{equation}
	\label{eq:sim_base_lifted}
	\begin{aligned}
		\dot{x}(t) &= B_{z(t)}\,x(t) + \sum_{k=1}^{K} w_k\,y_k(t) + f(t),\\
		\dot{y}_k(t) &= x(t) - r_k y_k(t), \qquad k=1,\dots,K,
	\end{aligned}
\end{equation}
where $x(t)\in\mathbb{R}^n$ is the network state and $\{y_k(t)\}_{k=1}^K$ are
auxiliary memory states induced by a nonnegative sum-of-exponentials (SOE)
approximation of a completely monotone Volterra kernel (Appendix~\ref{app:soe}).
The lifted state is $X(t)=(x(t),y_1(t),\dots,y_K(t))^\top$.

\paragraph{Regime switching}
The regime process $z(t)$ is a two-state continuous-time Markov chain with
state space
\[
z(t)\in\{S \ \text{(benign)},\, U \ \text{(unfavorable)}\},
\]
and exponentially distributed dwell times with rates $\lambda_{SU}$ and
$\lambda_{US}$. In particular, the exit rate from $U$ is $\lambda_U=\lambda_{US}$,
which is the quantity entering the theoretical tail index
$\lambda_U/\gamma_U$ in Subsection~\ref{subsec:quenched_tails}.

\paragraph{Instantaneous network operator}
The regime-dependent network operator takes the standard damped-coupling form
\begin{equation}
	\label{eq:Bz_def}
	B_z = -\gamma_z I + \beta_z W,
\end{equation}
where $W$ is a fixed directed adjacency matrix and $\gamma_z,\beta_z>0$.
Each regime is individually stable \emph{in the absence of memory}, i.e.,
\[
\gamma_z > \beta_z \rho(W),
\]
where $\rho(W)$ is the spectral radius. This isolates the role of memory and
regime persistence in generating extreme responses: any observed amplification
is a lifted-memory effect (conversion of stored stress into state growth) rather
than a trivial instantaneous instability of $B_z$.

\paragraph{Memory parameterization and heavy-tail regime}
The SOE parameters $\{(w_k,r_k)\}_{k=1}^K$ are chosen to approximate a slowly
decaying completely monotone kernel on $[0,T]$, with nonnegative weights $w_k$
and time scales $r_k$ spanning multiple orders of magnitude (Appendix~\ref{app:soe}).
This choice is made so that, in the unfavorable regime $U$ and in normal mode,
the lifted operator admits intervals of positive susceptibility in the sense of
\eqref{eq:gammaU_S_def}, i.e.\ an effective growth channel consistent with the
assumptions used in Subsection~\ref{subsec:quenched_tails}. Concretely, along
$U$-intervals we monitor the susceptibility indicator
\[
\mathcal{S}(t)=
\lambda_{\max}\!\left(
\frac{A_{z(t)}^{(m(t))}+(A_{z(t)}^{(m(t))})^\top}{2}
\right),
\]
and verify that $\mathcal{S}(t)$ becomes positive in the uncontrolled
memory-on configuration, thereby placing the system in the heavy-tail regime
predicted by the theory.

\paragraph{Forcing}
The external forcing is periodic and localized:
\begin{equation}
	f(t) = A \sin(\omega t)\,e_i,
\end{equation}
where $e_i$ excites a single node. This forcing provides persistent energy
injection while keeping the dynamics fully linear, allowing us to isolate
memory-induced tail effects under switching.

\paragraph{Observables}
The primary observable is the network energy
\begin{equation}
	E(t)=\|x(t)\|_2,
	\qquad
	B=\max_{t\in[0,T]} E(t),
\end{equation}
which defines the burst size used throughout.

\paragraph{Network topology}
The network topology used in all experiments is shown in
Figure~\ref{fig:network}. It is intentionally non-normal (directed), enabling
transient amplification when memory loading aligns the lifted state with
amplifying operator directions.

\begin{figure}[H]
	\centering
	\includegraphics[width=0.8\textwidth]{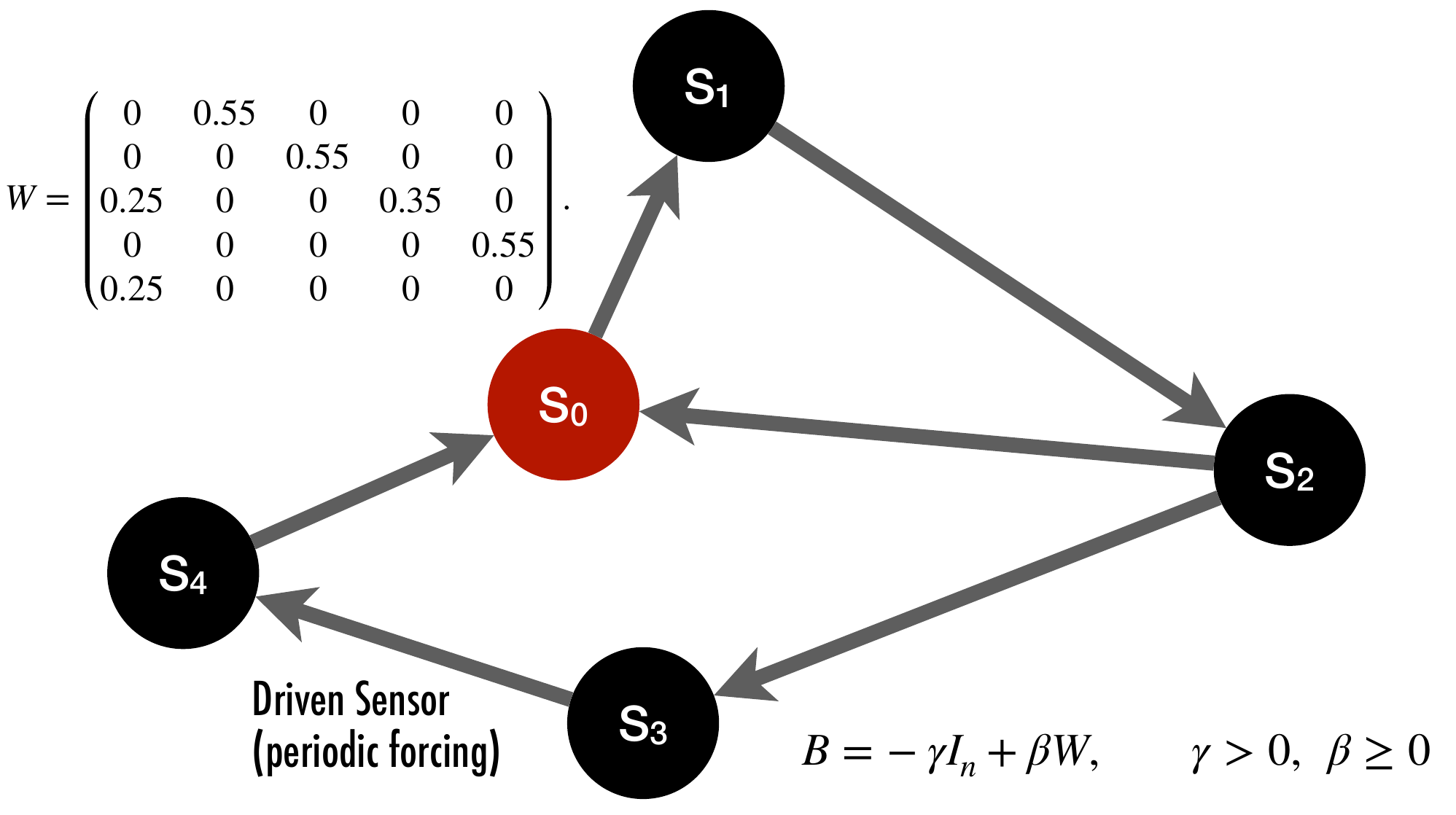}
	\caption{
		Illustrative directed sensor network used in the numerical experiments.
		Nodes represent sensing units, while directed weighted edges encode coupling strengths.
		The network dynamics are governed by
		\(
		B = -\gamma I + \beta W,
		\)
		where $\gamma>0$ controls intrinsic damping and $\beta$ scales inter-node coupling.
	}
	\label{fig:network}
\end{figure}

\subsection{Annealed stability, quenched amplification, and operator indicators}
\label{subsec:annealed_quenched_indicators}

Figure~\ref{fig:panel_memory_on} illustrates the central phenomenon of this work.
Despite individual regime stability and bounded annealed behavior, rare
quenched trajectories exhibit delayed but dramatic amplification.
This is the empirical signature of the annealed--quenched separation:
typical trajectories remain benign while a vanishing fraction dominates tail
risk on finite horizons.

\begin{figure}[H]
	\centering
	\includegraphics[width=\textwidth]{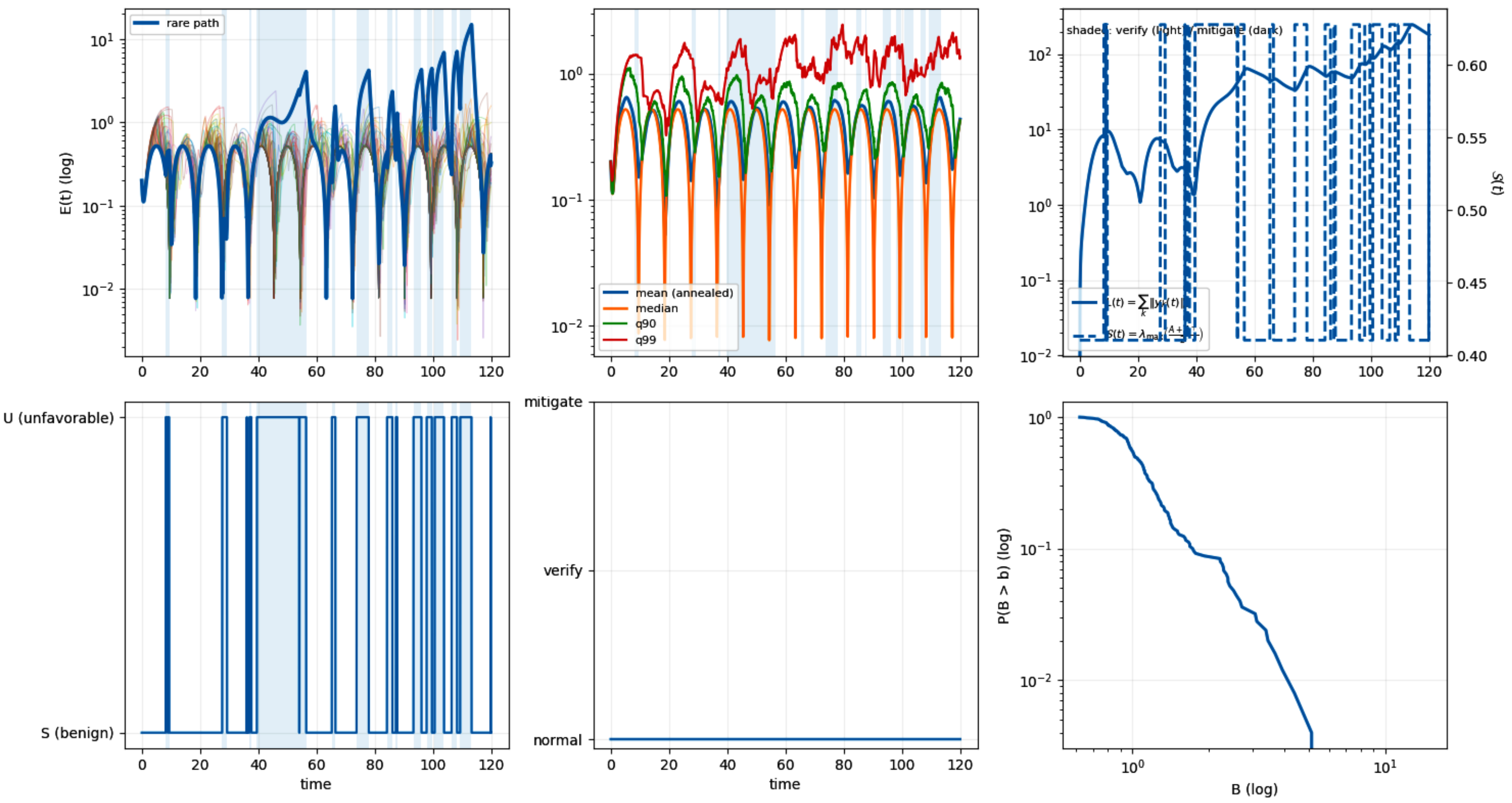}
	\caption{
		\textbf{Memory ON (no DDDAS): quenched amplification and tail risk under regime switching.}
		This figure summarizes an ensemble of Monte Carlo trajectories of the periodically forced networked system with SOE memory activated in the unfavorable regime.
		\textbf{(Top-left) Quenched sample paths.}
		Thin curves show $E(t)=\|x(t)\|_2$ for multiple realizations (same parameters, different regime paths), displayed on a logarithmic scale. The thick curve highlights the rare trajectory with the largest burst size $B=\max_{t\in[0,T]}E(t)$ in the ensemble. Shaded vertical bands indicate intervals where the regime is unfavorable ($z(t)=U$) along the rare trajectory. Despite stable-looking typical responses, rare paths exhibit delayed amplification driven by cumulative memory loading across repeated unfavorable episodes.
		\textbf{(Top-middle) Annealed statistics vs tail quantiles.}
		The ensemble mean (annealed) $E_{\text{mean}}(t)$, together with median and upper quantiles (e.g.\ $q_{0.9}$ and $q_{0.99}$), is shown on a log scale. The mean/median remain bounded and approximately periodic, while high quantiles drift upward, evidencing an annealed--quenched separation: typical behavior remains benign, yet extremes evolve differently.
		\textbf{(Top-right) Two indicators on the rare trajectory.}
		Left axis: latent memory load $L(t)=\sum_{k=1}^K\|y_k(t)\|_2$ (log scale), quantifying accumulated internal memory states in the SOE lifting. Right axis (dashed): operator susceptibility $\mathcal{S}(t)=\lambda_{\max}\!\left(\tfrac{A(t)+A(t)^\top}{2}\right)$ computed from the lifted operator. Peaks in $L(t)$ and elevated $\mathcal{S}(t)$ precede or coincide with amplification episodes, indicating entry into a dynamically dangerous configuration even when the nominal regime operator is individually stable.
		\textbf{(Bottom-left) Regime path for the rare trajectory.}
		The binary CTMC path $z(t)\in\{S,U\}$ corresponding to the highlighted rare realization. Importantly, large bursts need not be triggered by a single exceptionally long unfavorable dwell; instead, frequent unfavorable visits can accumulate through memory.
		\textbf{(Bottom-middle) Mode trace (here constant).}
		DDDAS mode $m(t)\in\{\text{normal},\text{verify},\text{mitigate}\}$ plotted for completeness; in this condition DDDAS is disabled, so the system remains in normal mode throughout.
		\textbf{(Bottom-right) Tail risk (CCDF of burst size).}
		Complementary cumulative distribution function $\mathbb{P}(B>b)$ of burst sizes across the ensemble (log--log axes), showing heavy-tailed behavior induced by the interaction of long memory and stochastic regime persistence.
	}
	\label{fig:panel_memory_on}
\end{figure}

The ensemble mean and median of $E(t)$ remain bounded and nearly periodic,
consistent with the finite-horizon annealed boundedness results in
Section~\ref{sec:math_results}. At the same time, the upper quantiles drift
and the maximal burst grows along rare realizations, in agreement with the
heavy-tail mechanism established in Subsection~\ref{subsec:quenched_tails}.
This provides a direct numerical counterpart of the theoretical statement that
annealed boundedness does not control quenched tails.

Along each realization we compute the two indicators introduced in
Subsection~\ref{subsec:observables_indicators}:
\[
L(t)=\sum_{k=1}^K \|y_k(t)\|_2,
\qquad
\mathcal{S}(t)=
\lambda_{\max}\!\left(
\frac{A_{z(t)}^{(m(t))}+(A_{z(t)}^{(m(t))})^\top}{2}
\right).
\]
As shown in Figure~\ref{fig:panel_memory_on} (top-right), both rise
systematically prior to amplification events. Notably, $\mathcal{S}(t)>0$
signals activation of an effective growth channel in the lifted operator
geometry (instantaneous Euclidean expansion), while $L(t)$ measures the
cumulative memory load that fuels delayed amplification across repeated
unfavorable episodes. This is precisely the operator-theoretic early-warning
signature exploited by the DDDAS policy.

\subsection{DDDAS intervention, baselines, and tail comparison}
\label{subsec:dddas_baselines_tail}

We next activate the DDDAS policy defined in
Subsection~\ref{subsec:dddas_policy}, using the operator-design principles of
Section~\ref{subsec:design_A}. The controller monitors $(L(t),\mathcal{S}(t))$
and modifies the lifted operator through verify/mitigate modes with hysteresis.
Crucially, the regime process $z(t)$ is not modified: the switching statistics
are exogenous and remain unchanged, so any tail reduction is attributable purely
to operator-level intervention in the lifted dynamics.

Figure~\ref{fig:panel_dddas} shows that adaptive intervention leaves annealed
statistics essentially unchanged, while strongly suppressing rare quenched
excursions. This is consistent with the pathwise boundedness and tail
truncation/exponent-improvement mechanisms established in
Subsection~\ref{subsec:dddas_policy}.

\begin{figure}[H]
	\centering
	\includegraphics[width=\textwidth]{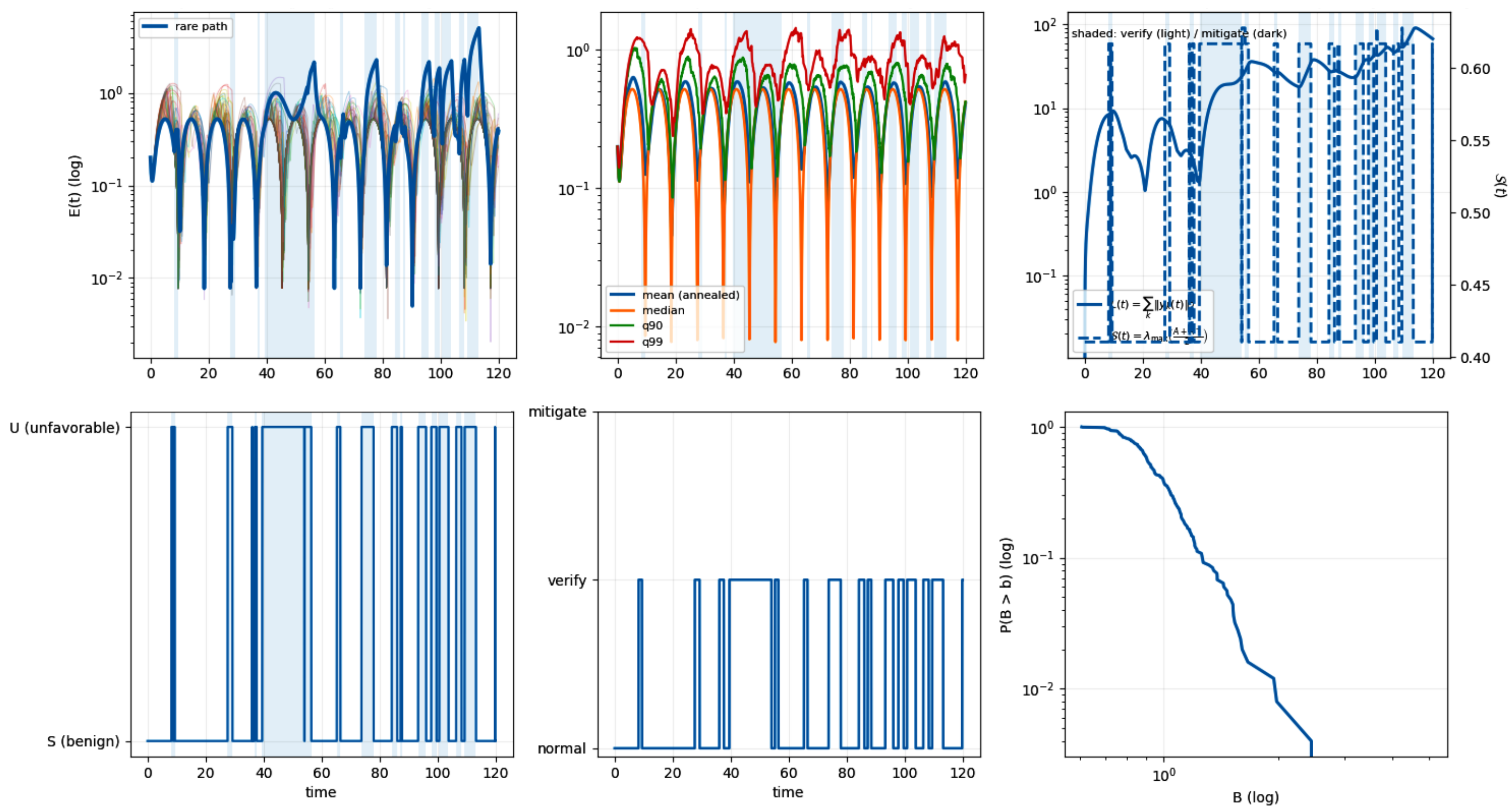}
	\caption{
		\textbf{Memory ON + DDDAS: tail-risk mitigation via a two-indicator hysteretic policy.}
		Same simulation setting as Fig.~\ref{fig:panel_memory_on}, but with DDDAS enabled. The policy monitors both the latent memory load $L(t)$ and the spectral susceptibility $\mathcal{S}(t)$, switching between normal/verify/mitigate modes with hysteresis (minimum mode duration) to avoid chattering.
		\textbf{(Top-left) Quenched sample paths.}
		Compared with the uncontrolled case, extreme quenched excursions are suppressed: the highlighted rare path achieves a substantially smaller maximum burst. Shaded bands again mark unfavorable-regime intervals for the rare trajectory.
		\textbf{(Top-middle) Annealed statistics vs tail quantiles.}
		Mean/median responses remain comparable (same forcing and regime statistics), while upper quantiles are pulled downward, showing that DDDAS primarily reduces extremes rather than distorting typical behavior.
		\textbf{(Top-right) Indicators and intervention windows (rare trajectory).}
		The latent load $L(t)$ (left axis, log scale) and susceptibility $\mathcal{S}(t)$ (right axis, dashed) increase during unfavorable periods, but when thresholds are crossed the controller enters verify or mitigate mode (background shading). Verify mode applies mild actions (e.g.\ reduced memory gain), whereas mitigate mode applies stronger actions (e.g.\ switching to a safer operator and partial memory discharge), preventing coherent accumulation of memory in the lifted states.
		\textbf{(Bottom-left) Regime path.}
		The CTMC regime path $z(t)$ for the rare trajectory; DDDAS does not alter regime switching and therefore preserves the statistics of $z(t)$.
		\textbf{(Bottom-middle) DDDAS mode trace.}
		The controller mode $m(t)$ over time (normal/verify/mitigate). Mode transitions occur primarily during unfavorable-regime intervals and are constrained by hysteresis, yielding sustained but non-aggressive interventions rather than rapid toggling.
		\textbf{(Bottom-right) Tail risk (CCDF of burst size).}
		The empirical CCDF $\mathbb{P}(B>b)$ across the ensemble shifts downward relative to Fig.~\ref{fig:panel_memory_on}, indicating a reduction in the probability of rare high-impact bursts while keeping the same regime process and forcing.
	}
	\label{fig:panel_dddas}
\end{figure}

We compare the DDDAS-controlled system with two baselines:
(i) a memory-free dynamics obtained by setting $w_k\equiv 0$, and
(ii) a SAFE-in-U system where the unfavorable regime is permanently replaced by
a strongly damped operator. Figures~\ref{fig:panel_nomem} and~\ref{fig:panel_safe}
show that both baselines suppress extreme events, but they do so by either
eliminating memory altogether (removing the diagnostic channel) or by modifying
the unfavorable regime dynamics structurally. In contrast, the proposed DDDAS
strategy preserves memory as an early-warning signal while selectively
truncating tail risk through targeted operator modifications.

\begin{figure}[H]
	\centering
	\includegraphics[width=\textwidth]{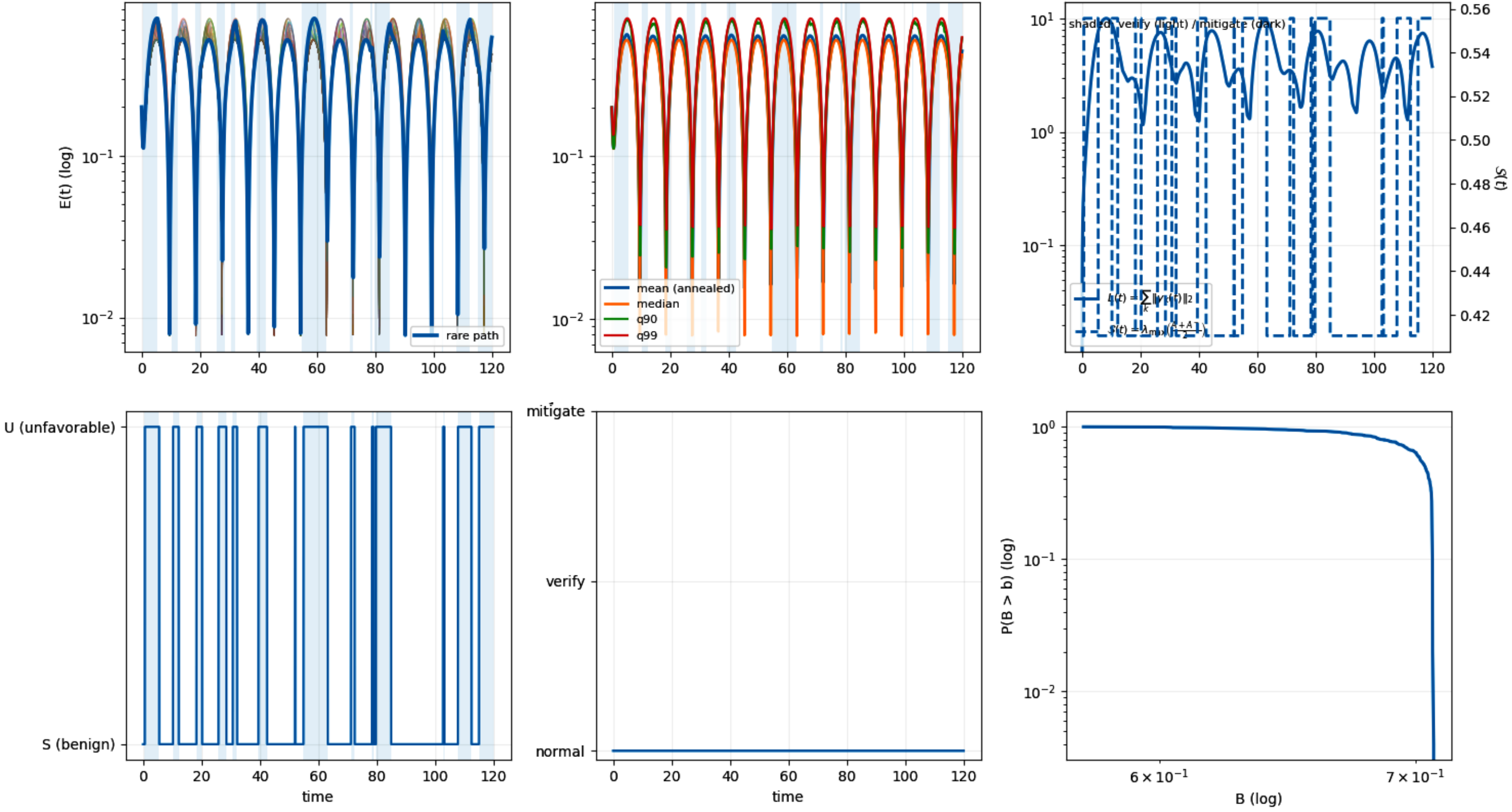}
	\caption{
		\textbf{Memory OFF (control): regime switching alone does not generate heavy tails.}
		Same network, forcing, and regime switching as in Figs.~\ref{fig:panel_memory_on}--\ref{fig:panel_dddas}, but with memory gain set to zero (no SOE memory contribution).
		\textbf{(Top-left) Quenched sample paths.}
		All realizations of $E(t)=\|x(t)\|_2$ remain tightly clustered around a bounded oscillatory response on the log scale, and the maximum-burst trajectory is not qualitatively distinct from typical paths, indicating the absence of delayed amplification.
		\textbf{(Top-middle) Annealed statistics vs tail quantiles.}
		Mean, median, and upper quantiles remain close and stationary over time, demonstrating that typical and extreme behaviors coincide in the memoryless system.
		\textbf{(Top-right) Indicators.}
		The latent load $L(t)$ remains near zero (up to numerical integration effects), and the spectral susceptibility $\mathcal{S}(t)$ reflects only the instantaneous geometry of the memoryless operator. Without memory states to accumulate, sustained increases in indicators do not occur and there is no pathway to coherent amplification.
		\textbf{(Bottom-left) Regime path.}
		The regime path $z(t)$ for the highlighted realization is shown for reference; despite unfavorable intervals, the dynamics remain bounded because there is no historical accumulation mechanism.
		\textbf{(Bottom-middle) Mode trace (here constant).}
		DDDAS is disabled in this control run and the system remains in normal mode.
		\textbf{(Bottom-right) Tail risk (CCDF of burst size).}
		The CCDF $\mathbb{P}(B>b)$ exhibits a sharp cutoff and decays rapidly, confirming that long memory is the essential ingredient that converts regime persistence into heavy-tailed burst statistics.
	}
	\label{fig:panel_nomem}
\end{figure}

\begin{figure}[H]
	\centering
	\includegraphics[width=\textwidth]{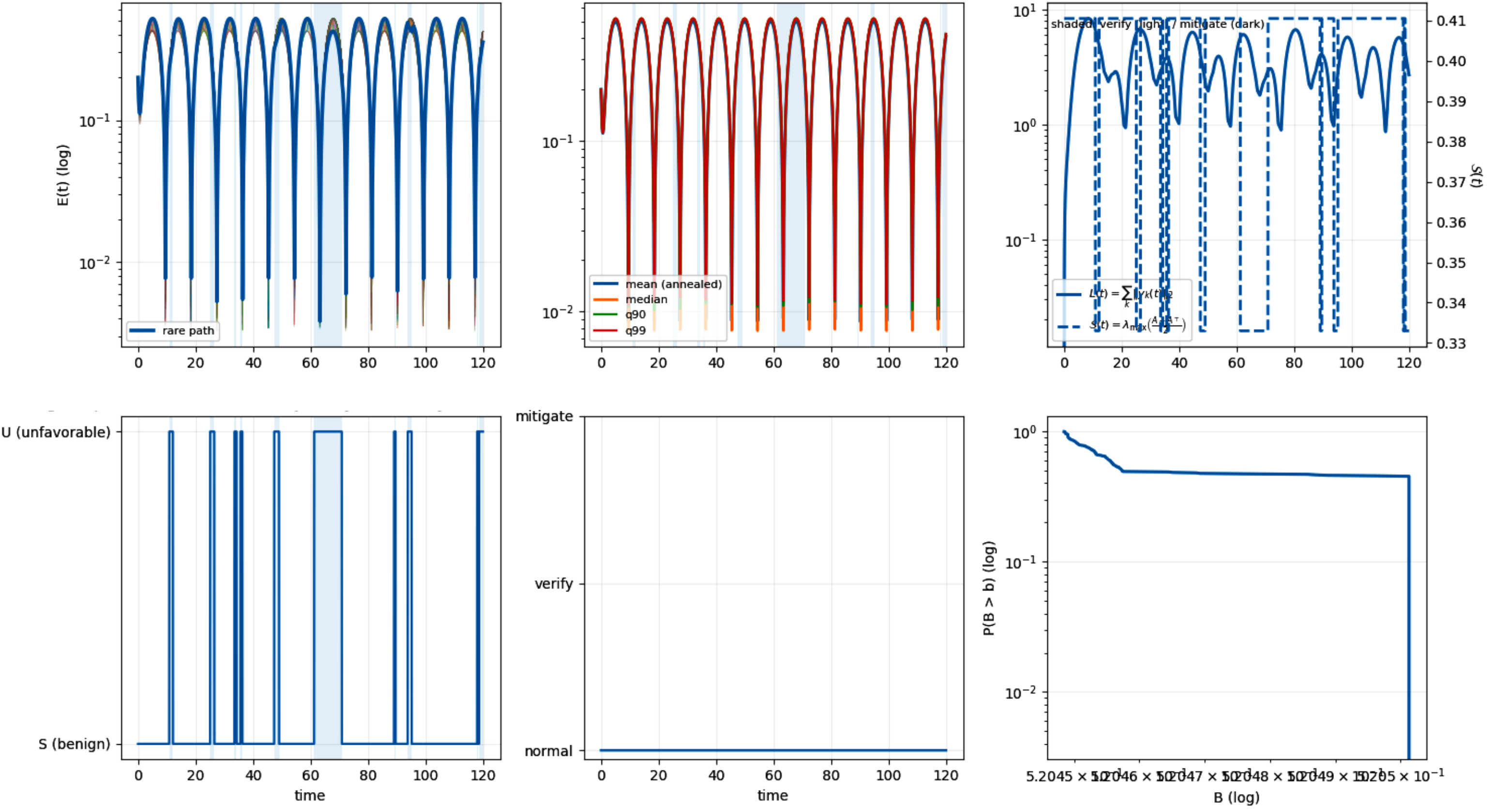}
	\caption{
		SAFE-in-U baseline, where the unfavorable regime is permanently replaced by a strongly damped operator.
		\textbf{Top-left:} Quenched trajectories of $E(t)$ show no amplification.
		\textbf{Top-middle:} Annealed statistics remain bounded and periodic.
		\textbf{Top-right:} Diagnostic indicators remain below critical thresholds, reflecting structural over-damping.
		\textbf{Bottom-left:} Regime path $z(t)$, shown for reference.
		\textbf{Bottom-middle:} DDDAS action mode (inactive).
		\textbf{Bottom-right:} CCDF of burst size $B$, indicating strong tail suppression achieved by modifying regime dynamics.
	}
	\label{fig:panel_safe}
\end{figure}

Figure~\ref{fig:ccdf_comparison} compares the empirical CCDFs of $B$ across all
scenarios. In the uncontrolled memory-on case, the CCDF exhibits a clear
power-law regime over an intermediate scale range, as predicted by the theory.
Under DDDAS control, the empirical CCDF shifts downward and typically steepens;
in the strongly contractive mitigation configuration, the tail exhibits visible
finite-horizon truncation. These outcomes match the mathematical predictions:
DDDAS reduces the effective unfavorable growth rate (or enforces contraction-on-demand),
thereby increasing the effective tail exponent or truncating tails on $[0,T]$.

\begin{figure}[H]
	\centering
	\includegraphics[width=0.75\textwidth]{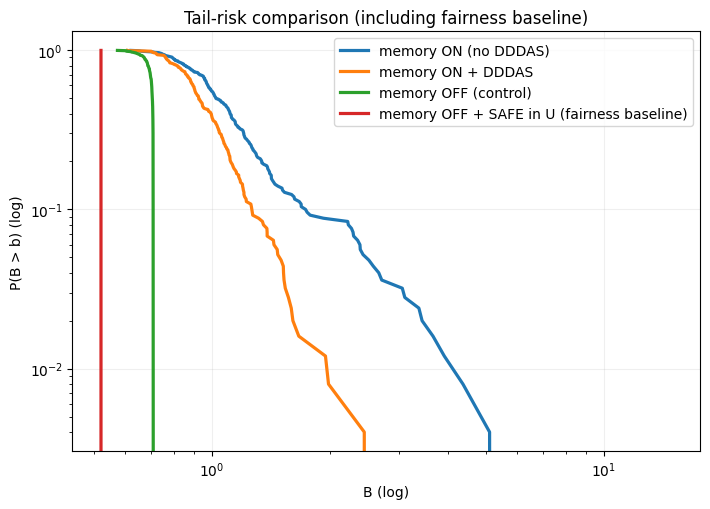}
	\caption{
		Comparison of burst-size tail distributions across control strategies.
		Shown are the CCDFs of $B=\max_t E(t)$ for memory ON without DDDAS, memory ON with DDDAS, memory OFF, and the SAFE-in-U baseline.
		The DDDAS-controlled system achieves tail-risk reduction comparable to the SAFE-in-U baseline while preserving the original regime statistics, outperforming memory-free dynamics.
	}
	\label{fig:ccdf_comparison}
\end{figure}

\paragraph{Empirical tail indices}
The slope of the log--log CCDF provides an empirical estimate $\widehat{\alpha}$
of the tail index on the range $b\ge b_{\min}$, with $b_{\min}$ chosen via a high-quantile
rule and a stability check (Appendix~\ref{app:ccdf_fitting}).
We compare $\widehat{\alpha}$ with the theoretical prediction $\lambda_U/\gamma_U$
from Subsection~\ref{subsec:quenched_tails}, where $\gamma_U$ is estimated from
uncontrolled unfavorable intervals using the operator susceptibility and dwell-level growth diagnostics described in
Appendix~\ref{app:gamma_estimation}. Under DDDAS control, the effective growth rate
in $U$ is reduced, yielding a steeper empirical slope or finite-horizon truncation,
consistent with the theory in Subsection~\ref{subsec:dddas_policy}.

\paragraph{Summary}
Collectively, these experiments establish: (i) long memory and regime persistence
generate heavy-tailed burst statistics despite annealed boundedness; (ii) the lifted
indicators $(L(t),\mathcal{S}(t))$ provide reliable early warning; (iii) DDDAS
intervention truncates or steepens tails without altering the exogenous regime
statistics; and (iv) memory, when properly monitored, is a diagnostic asset rather
than a liability.

\section{Discussion and Conclusions}
\label{sec:discussion}

This paper isolates a structural mechanism for rare, high-impact transient events in
linear networked systems operating under intermittent adverse conditions and long memory.
Working entirely within a finite-dimensional lifted switched ODE representation,
we exhibit a sharp separation between \emph{annealed boundedness}
(mean stability on finite horizons)
and \emph{quenched tail risk}
(trajectory-level heavy-tailed bursts).
The analysis identifies when heavy tails are not anomalies but emergent consequences of
(i) regime persistence,
(ii) memory-induced amplification geometry,
and (iii) alignment effects in the lifted state space.

\subsection*{Relation to existing heavy-tail and switching-system theories}

A natural referee objection is that power-law tails under switching are well known in the
Kesten--Goldie paradigm and in Markov-modulated linear systems.
Indeed, heavy tails have been rigorously established for discrete-time stochastic
recursions and continuous-time diffusions under regime switching
\cite{kesten1973,goldie1991,buraczewski2016,desaporta2005,deSaporta2004tails}.

The contribution of the present work is not the existence of heavy tails per se,
but the identification of a \emph{deterministic, memory-driven, operator-geometric mechanism}
for quenched amplification in a forced continuous-time system where randomness enters
\emph{only} through the regime path.

\paragraph{Continuous-time, forced, memory-induced mechanism}
Unlike classical multiplicative-noise models or random matrix products
\cite{kesten1973,goldie1991,buraczewski2016},
the present system is linear, deterministic between regime switches,
and driven by bounded forcing.
Heavy tails arise because long memory, embedded through a Volterra-to-ODE lifting,
creates effective growth channels in the lifted operator during unfavorable dwells.
This places the mechanism closer in spirit to Markov-modulated linear dynamics
\cite{deSaporta2004tails,desaporta2005},
but with amplification generated endogenously by memory geometry rather than by
instability of instantaneous dynamics.

\paragraph{Operator-theoretic tail exponent and measurable diagnostics}
A central distinction is that the tail exponent $\lambda_U/\gamma_U$
is not a fitted phenomenological parameter.
The dwell-time rate $\lambda_U$ is a property of the exogenous regime process,
while $\gamma_U$ is an operator-theoretic quantity given by the Euclidean matrix measure
(logarithmic norm) of the lifted operator.
This directly connects rigorous tail bounds to measurable diagnostics,
bridging abstract heavy-tail theory
\cite{deSaporta2004tails,buraczewski2016}
with nonmodal growth analysis
\cite{trefethen1993,schmid2007,embree2005}.

\paragraph{Tail shaping without altering regime statistics}
Most existing approaches to tail mitigation rely either on modifying the switching process
or on static over-damping.
In contrast, the proposed DDDAS policy leaves the regime process untouched
and instead performs \emph{contraction on demand} by switching operators
only when the system enters a dynamically dangerous region of lifted state space.
This distinguishes the approach from classical stabilization
\cite{mao2007stochastic,yuan2010stability}
and aligns it with the DDDAS philosophy of adaptive, state-aware intervention
\cite{darema2004dddas,darema2009}.

\subsection*{Lifted-operator geometry as the organizing principle}

A second likely objection is that the susceptibility indicator $\mathcal S(t)$
is merely a reformulation of classical non-normal transient-growth diagnostics.
Indeed, symmetric-part growth rates and logarithmic norms are standard tools
in nonmodal stability theory
\cite{trefethen1993,schmid2007,embree2005}.

The key point here is that in systems with memory,
the relevant geometry does not live in the original network state space,
but in the lifted space induced by the Volterra embedding.
Even when each instantaneous network operator $B_z$ is stable,
the lifted operator $A_z^{(0)}$ can admit positive susceptibility windows
due to accumulated memory states.
Thus, the appropriate early-warning object is not the spectrum of $B_z$,
but the geometry of the lifted operator family $\{A_z^{(m)}\}$.

\subsection*{Heavy tails as a balance of time scales and geometry}

The power-law lower bound derived in Section~\ref{sec:math_results}
formalizes a simple but structurally important message:
quenched extremes are governed by a competition between
the unfavorable dwell-time tail
$\tau^{(U)}\sim\mathrm{Exp}(\lambda_U)$
and the effective lifted amplification rate $\gamma_U$.
This mirrors the balance between persistence and growth
characteristic of Markov-modulated systems
\cite{deSaporta2004tails,desaporta2005},
but with $\gamma_U$ emerging from memory-induced operator geometry rather than
from eigenvalue instability.

\subsection*{DDDAS as contraction-on-demand tail shaping}

The control contribution should not be interpreted as a classical stabilization result.
Instead, DDDAS acts as a \emph{tail-shaping} mechanism:
it selectively disrupts rare amplification pathways by enforcing temporary contraction
only when indicators $(L(t),\mathcal S(t))$ signal elevated tail risk.
This allows the policy to reduce extreme-event probability
while leaving annealed statistics largely unchanged,
a behavior that static control or memory elimination cannot replicate.

\subsection*{Assumptions, empirical verifiability, and robustness}

\paragraph{Cone condition and empirical testability}
The cone (alignment) condition used in the heavy-tail lower bound
encodes persistence of an amplifying direction along rare burst-producing segments.
Rather than being purely technical, it is empirically testable:
the alignment ratio
$a(t)=(v_U^\top X(t))/\|X(t)\|_2$
can be monitored and summarized along unfavorable dwells.
This makes the heavy-tail mechanism falsifiable in simulations and, in principle,
in data-driven applications.

\paragraph{Finite-horizon focus}
All results are finite-horizon by design.
This matches operational risk assessment, where the relevant observable is the burst
$B_T=\sup_{t\le T}\|x(t)\|_2$.
The theory explicitly identifies when finite-horizon truncation enters,
and the numerical results illustrate the resulting curvature or cutoff in CCDFs.

\paragraph{SOE lifting and approximation error}
The sum-of-exponentials (SOE) approximation introduces a numerical modeling layer.
By enforcing positivity via NNLS and monitoring kernel approximation error,
we ensure that the lifted dynamics faithfully represent completely monotone memory kernels,
consistent with approximation theory
\cite{beylkinMonzon2005expSums,gripenberg1990volterra,pruss1993evolution}.

\subsection*{Outlook}

Natural extensions include multiple unfavorable regimes,
interacting growth channels, and nonlinear generalizations
where memory states couple to state-dependent switching.
From an applied perspective, joint estimation of
$\lambda_U$ and $\gamma_U$ from data streams
offers a principled route to real-time tail-risk indicators
and adaptive control grounded in operator geometry.

\paragraph{Conclusion}
Systems with memory can be \emph{stable on average} yet \emph{risky in the tails}.
By exposing a concrete lifted-geometry mechanism for quenched amplification,
and by providing computable diagnostics together with a provably tail-reducing
contraction-on-demand policy,
this work contributes a physically interpretable and operationally actionable
framework for tail risk in regime-switching systems with memory.

\appendix

\section{Computational and Numerical Details}
\label{app:computational}

This appendix provides implementation-level details supporting the numerical
experiments in Section~\ref{sec:numerics}. The emphasis is on reproducibility and on
a transparent bridge between the empirical measurements and the theoretical quantities
introduced in Section~\ref{sec:math_results}, in particular Subsection~\ref{subsec:quenched_tails}:
(i) the SOE-based lifted embedding and its numerical construction,
(ii) empirical estimation of the effective unfavorable growth rate $\gamma_U$
(or its cone-restricted counterpart when alignment restrictions are enforced), and
(iii) fitting of burst-size tails via CCDF slopes and comparison with the predicted
tail index $\lambda_U/\gamma_U$.

\subsection{SOE approximation and lifted embedding}
\label{app:soe}

\paragraph{Target kernel class}
In each regime, the memory kernel is assumed completely monotone on $[0,\infty)$ and hence
admits the Bernstein representation
\begin{equation}\label{eq:bernstein_app}
	g(t)=\int_0^\infty e^{-rt}\,\mu(dr),
\end{equation}
for a positive measure $\mu$. On the finite horizon $[0,T]$, we approximate the integral by a finite positive SOE quadrature
\begin{equation}\label{eq:soe_app}
	g(t)\approx g^{(K)}(t):=\sum_{k=1}^K w_k e^{-r_k t},
	\qquad w_k>0,\ r_k>0.
\end{equation}

\paragraph{Choice of nodes}
The decay rates $\{r_k\}$ are selected on a logarithmic grid spanning relevant memory time scales:
\begin{equation}\label{eq:rk_grid_app}
	r_k=r_{\min}\left(\frac{r_{\max}}{r_{\min}}\right)^{\frac{k-1}{K-1}},
	\qquad k=1,\dots,K,
\end{equation}
with $r_{\min}\approx T^{-1}$ and $r_{\max}$ chosen so that $e^{-r_{\max}\Delta t}$ is numerically negligible at the solver time step
$\Delta t$.

\paragraph{Positive weights via constrained fitting (NNLS)}
Weights $\{w_k\}$ are obtained by nonnegative least squares fitting on a design grid $\{t_j\}_{j=1}^J\subset[0,T]$:
\begin{equation}\label{eq:nnls_app}
	\min_{w_k\ge 0}\ \sum_{j=1}^J\Big(g(t_j)-\sum_{k=1}^K w_k e^{-r_k t_j}\Big)^2.
\end{equation}
Positivity preserves complete monotonicity at the discrete level and prevents spurious sign-changing memory gains in the lifted dynamics.

\paragraph{Lifted ODE used in simulations}
With the SOE approximation, the Volterra term admits the standard Markovian embedding
\begin{align}\label{eq:lifted_app_system}
	\dot{x}(t) &= B_{z(t)}x(t)+\sum_{k=1}^K w_k y_k(t)+f(t),\\
	\dot{y}_k(t) &= x(t)-r_k y_k(t),\qquad k=1,\dots,K.
\end{align}
This matches the normal-mode lifted structure used throughout the manuscript (cf.\ \eqref{eq:lifted_structure}).

\paragraph{Integration, event handling, and accuracy checks}
We integrate \eqref{eq:lifted_app_system} with an explicit adaptive solver (e.g.\ RK45) using event handling for CTMC jump times
(and, when DDDAS is enabled, for mode switching times). SOE accuracy is monitored by the relative kernel error on $[0,T]$,
\begin{equation}\label{eq:soe_error_metric}
	\varepsilon_{\mathrm{rel}}
	:=
	\sup_{t\in[0,T]}\frac{|g(t)-g^{(K)}(t)|}{\max\{1,|g(t)|\}},
\end{equation}
and by convergence of key observables (e.g.\ fitted CCDF slopes) under refinement of $K$ and the time-scale coverage
(Appendix~\ref{app:sens_soe}).

\subsection{Estimating the unfavorable effective growth rate $\gamma_U$}
\label{app:gamma_estimation}

The tail mechanism in Subsection~\ref{subsec:quenched_tails} depends on a competition between
(i) unfavorable dwell-time persistence $\tau^{(U)}\sim\mathrm{Exp}(\lambda_U)$ and
(ii) an effective lifted growth rate $\gamma_U>0$ during $U$-dwells. In the lifted ODE formulation, $\gamma_U$ is naturally linked to the
Euclidean matrix measure (logarithmic norm), i.e.\ to the susceptibility indicator $\mathcal{S}(t)$.

We report three complementary estimators, from operator-based to trajectory-based.

\paragraph{(A) Operator-based estimator via the matrix measure}
When $z(t)=U$ and the mode is normal ($m(t)=0$), the susceptibility is constant and equals
\begin{equation}\label{eq:gamma_operator_app}
	\gamma_{U,\mathrm{op}}
	:=
	\mu_2(A_U^{(0)})
	=
	\lambda_{\max}\!\left(\frac{A_U^{(0)}+(A_U^{(0)})^\top}{2}\right).
\end{equation}
This quantity is computed directly from the lifted matrix $A_U^{(0)}$.

\paragraph{(B) Cone-restricted effective rate (alignment diagnostic)}
If a cone/alignment restriction is used in the heavy-tail argument (as discussed in Subsection~\ref{subsec:quenched_tails}),
let $v_U$ denote a unit eigenvector attaining \eqref{eq:gamma_operator_app} and define the alignment ratio
\begin{equation}\label{eq:alignment_ratio_app}
	a(t):=\frac{v_U^\top X(t)}{\|X(t)\|_2}\in[-1,1].
\end{equation}
For a chosen level $\alpha\in(0,1]$, restrict to times where $z(t)=U$, $m(t)=0$, and $a(t)\ge \alpha$, and estimate
\begin{equation}\label{eq:gamma_cone_app}
	\widehat{\gamma}_{U,\mathrm{cone}}
	:=
	\mathrm{Quantile}_{q}\left(
	\frac{1}{\Delta}\log\frac{\|X(t+\Delta)\|_2}{\|X(t)\|_2}
	\ \Big|\ z(t)=U,\ m(t)=0,\ a(t)\ge \alpha
	\right),
\end{equation}
where $\Delta$ is a short window and $q$ is taken high (e.g.\ $q=0.9$), reflecting that extreme bursts are driven by the most amplifying windows.

\paragraph{(C) Dwell-level estimator}
For each realization, identify maximal intervals $[t_0,t_1]$ such that $z(t)=U$ and $m(t)=0$ for all $t\in[t_0,t_1]$, and compute
\begin{equation}\label{eq:gamma_dwell_app}
	\widehat{\gamma}_{U,\mathrm{dwell}}
	=
	\frac{1}{t_1-t_0}
	\log\!\left(\frac{\|X(t_1)\|_2}{\|X(t_0)\|_2}\right).
\end{equation}
Only intervals with $t_1-t_0\ge \Delta_{\min}$ are retained. We then summarize $\gamma_U$ by a high quantile of the empirical distribution.

\paragraph{Forcing projection diagnostic}
When a forcing projection assumption is invoked in theory, we validate it numerically by monitoring $v_U^\top \widetilde f(t)$
(or its discrete-time minimum over the sampled grid) to ensure it remains positive or at least non-negligible on the reported horizons.

\subsection{CCDF fitting, tail-index estimation, and theory comparison}
\label{app:ccdf_fitting}

\paragraph{Burst observable and empirical CCDF}
For each trajectory we compute
\[
B_T=\max_{t\in[0,T]}\|x(t)\|_2,
\]
and estimate the empirical CCDF $\widehat{\mathbb{P}}(B_T>b)$ from an ensemble of $N$ independent realizations.

\paragraph{Choice of tail threshold}
A lower cutoff $b_{\min}$ is selected using either:
(i) a fixed upper-quantile rule (e.g.\ the $0.9$-quantile of $\{B_T\}$), or
(ii) a stability rule: choose the smallest $b_{\min}$ for which the fitted log--log slope is stable under small perturbations of $b_{\min}$.

\paragraph{Log--log slope fit and uncertainty quantification}
On $b\ge b_{\min}$ we fit the standard power-law surrogate
\begin{equation}\label{eq:ccdf_slope_fit_app}
	\log \widehat{\mathbb{P}}(B_T>b)\approx -\widehat{\alpha}\log b + \widehat{c}
\end{equation}
by least squares. Confidence intervals are obtained via bootstrap resampling of trajectories with replacement.

\paragraph{Theoretical comparison}
We compare $\widehat{\alpha}$ against
\begin{equation}\label{eq:tail_index_compare_app}
	\alpha_{\mathrm{th}}=\frac{\lambda_U}{\gamma_U},
\end{equation}
where $\gamma_U$ is taken as $\gamma_{U,\mathrm{op}}$, $\widehat{\gamma}_{U,\mathrm{cone}}$, or a conservative hybrid
(e.g.\ $\gamma_U=\min\{\gamma_{U,\mathrm{op}},\widehat{\gamma}_{U,\mathrm{cone}}\}$).
Finite-horizon effects emerge when $(1/\gamma_U)\log b$ approaches the horizon scale used in the theoretical lower bound; in that case the CCDF
may exhibit curvature or truncation and the fitted slope becomes scale-dependent.

\subsection{Reporting protocol}
\label{app:reporting}

For each experiment, we report:
\begin{itemize}
	\item \textbf{Regime process:} $(\lambda_{SU},\lambda_{US})$ and the induced $\lambda_U$.
	\item \textbf{Network operator:} $(n,\rho(W))$, coupling parameters $(\gamma_z,\beta_z)$, and the forced node(s).
	\item \textbf{Memory discretization:} $(K,r_{\min},r_{\max})$ and the kernel error metric $\varepsilon_{\mathrm{rel}}$.
	\item \textbf{Growth rates:} $\gamma_{U,\mathrm{op}}$, dwell-based quantiles of $\widehat{\gamma}_{U,\mathrm{dwell}}$,
	and (if used) $\widehat{\gamma}_{U,\mathrm{cone}}$ with the chosen cone level $\alpha$.
	\item \textbf{Tail statistics:} fitted slope $\widehat{\alpha}$ with CIs, and comparison with $\alpha_{\mathrm{th}}=\lambda_U/\gamma_U$.
\end{itemize}
This protocol suffices to reproduce the figures in Section~\ref{sec:numerics} and to validate the reported scaling.

\section{Sensitivity Analysis}
\label{app:sensitivity}

This appendix assesses robustness with respect to modeling, numerical, and control parameters. We distinguish \emph{structural} effects
(changing the presence/absence of heavy tails) from \emph{quantitative} effects (changing prefactors or truncation scales).

\subsection{Sensitivity to SOE resolution and kernel time-scale coverage}
\label{app:sens_soe}

We repeat experiments for increasing $K$ and for varied $(r_{\min},r_{\max})$ in \eqref{eq:rk_grid_app}, keeping the regime process and network topology fixed.
We report: (i) convergence of $\varepsilon_{\mathrm{rel}}$ in \eqref{eq:soe_error_metric}, (ii) stability of $\gamma_{U,\mathrm{op}}$ and of the distribution of
$\widehat{\gamma}_{U,\mathrm{dwell}}$, and (iii) stability of the fitted CCDF slope $\widehat{\alpha}$.
Empirically, once dominant memory time scales are resolved, $\widehat{\alpha}$ varies only within statistical uncertainty.

\subsection{Sensitivity to regime statistics and dwell-time persistence}
\label{app:sens_regime}

We vary $(\lambda_{SU},\lambda_{US})$ while keeping all other parameters fixed.
The theory predicts $\alpha_{\mathrm{th}}=\lambda_U/\gamma_U$, hence $\widehat{\alpha}$ should increase approximately linearly with $\lambda_U$ provided finite-horizon truncation is negligible.
We report: (i) fitted $\widehat{\alpha}$ versus $\lambda_U$, (ii) stability of $\gamma_U$ estimates across runs, and (iii) annealed statistics to confirm that changes affect tails rather than typical behavior.

\subsection{Sensitivity to network topology, non-normality, and coupling strength}
\label{app:sens_network}

We compare directed networks with similar spectral radius $\rho(W)$ but different non-normality levels, while maintaining regime-wise stability
$\gamma_z>\beta_z\rho(W)$.
Heavy tails correlate with positive susceptibility $\gamma_{U,\mathrm{op}}>0$ (or elevated $\mathcal{S}(t)$ during $U$ windows), rather than with eigenvalue stability alone.

\subsection{Sensitivity to forcing amplitude, frequency, and projection}
\label{app:sens_forcing}

We vary forcing amplitude $A$ and frequency $\omega$ in $f(t)=A\sin(\omega t)e_i$.
Consistent with the mechanism in Section~\ref{sec:math_results}, the tail \emph{exponent} is controlled primarily by $\lambda_U$ and $\gamma_U$, while forcing affects prefactors and truncation scales.
We additionally report forcing projection diagnostics when a projection hypothesis is used in the theoretical argument.

\subsection{Sensitivity to DDDAS thresholds, hysteresis, and mitigation strength}
\label{app:sens_dddas}

We vary threshold levels $(\tau_L^{(1)},\tau_L^{(2)},\tau_S^{(1)},\tau_S^{(2)})$ and the minimum dwell time $\Delta_{\min}$.
We report: (i) intervention rate, (ii) change in annealed statistics (distortion of typical behavior), and (iii) change in tail behavior (slope $\widehat{\alpha}$ or evidence of truncation).
Moderate thresholds reduce the effective unfavorable growth rate from $\gamma_U$ to $\gamma_U^{\mathrm{eff}}<\gamma_U$, yielding predictable tail improvement while preserving regime statistics.

\subsection{Summary of robustness}
\label{app:sens_summary}

Overall, the sensitivity analysis supports three conclusions:
(i) quenched heavy tails are structurally induced by the interaction of memory and dwell-time persistence,
(ii) the predicted tail index $\lambda_U/\gamma_U$ is stable under numerical refinement once memory time scales are resolved, and
(iii) DDDAS mitigation improves tail behavior robustly across a broad parameter range while largely preserving annealed statistics.

\section{Implementation Pipeline (NNLS, dwell detection, CCDF slope fit)}
\label{app:implementation}

This appendix section summarizes the practical implementation pipeline used to generate the numerical results:
SOE fitting via NNLS, detection of uncontrolled unfavorable dwells, estimation of $\gamma_U$, selection of $b_{\min}$, and CCDF slope fitting with bootstrap uncertainty quantification.

\begin{algorithm}[H]
	\caption{Implementation pipeline for SOE fitting, dwell detection, and CCDF tail-index estimation}
	\label{alg:implementation}
	\small
	
	\KwIn{
		Kernel samples $\{(t_j,g(t_j))\}_{j=1}^J$; log-grid $\{r_k\}_{k=1}^K$; horizon $T$; Monte Carlo size $N$;
		CTMC rates $(\lambda_{SU},\lambda_{US})$; minimum dwell $\Delta_{\min}$; bootstrap size $B$; quantile levels $q_\gamma,q_b$
	}
	\KwOut{
		NNLS SOE weights $\{w_k\}$; effective growth $\widehat{\gamma}_U$; tail index $\widehat{\alpha}$ with CI; cutoff $b_{\min}$; theoretical index $\alpha_{\mathrm{th}}$
	}
	
	\BlankLine
	\textbf{(SOE by NNLS)} Construct $G\in\mathbb{R}^{J\times K}$ with $G_{jk}\gets e^{-r_k t_j}$ and solve
	$w^\star \in \arg\min_{w\ge 0}\sum_{j=1}^J\Big(g(t_j)-\sum_{k=1}^K G_{jk}w_k\Big)^2$; set $w_k\gets w_k^\star$\;
	
	\BlankLine
	Initialize lists $\mathcal{B}\gets\emptyset$ (bursts), $\mathcal{G}\gets\emptyset$ (dwell growth rates)\;
	\For{$n\gets 1$ \KwTo $N$}{
		Sample CTMC path $z^{(n)}(t)$ on $[0,T]$; generate mode $m^{(n)}(t)$; integrate \eqref{eq:sim_base_lifted}\;
		$B_T^{(n)}\gets \max_{t\in[0,T]}\|x^{(n)}(t)\|_2$; append $B_T^{(n)}$ to $\mathcal{B}$\;
		Extract maximal intervals $\{[t_0,t_1]\}$ where $z^{(n)}(t)=U$ and $m^{(n)}(t)=0$\;
		\ForEach{$[t_0,t_1]$}{
			\If{$t_1-t_0\ge \Delta_{\min}$}{
				$\widehat{\gamma}\gets \frac{1}{t_1-t_0}\log\!\big(\|X^{(n)}(t_1)\|_2/\|X^{(n)}(t_0)\|_2\big)$\;
				append $\widehat{\gamma}$ to $\mathcal{G}$\;
			}
		}
	}
	
	\BlankLine
	$\widehat{\gamma}_U \gets \mathrm{Quantile}_{q_\gamma}(\mathcal{G})$\;
	$b_{\min}\gets \mathrm{Quantile}_{q_b}(\mathcal{B})$; refine $b_{\min}$ until fitted log--log slope over $b\ge b_{\min}$ is stable\;
	Fit $\log \widehat{\mathbb{P}}(B_T>b)=c-\widehat{\alpha}\log b$ by least squares on $b\ge b_{\min}$\;
	
	\BlankLine
	Initialize $\mathcal{A}\gets\emptyset$\;
	\For{$b\gets 1$ \KwTo $B$}{
		Resample $\mathcal{B}^{(b)}$ from $\mathcal{B}$ with replacement; refit $\widehat{\alpha}^{(b)}$; append to $\mathcal{A}$\;
	}
	Return CI as quantiles of $\mathcal{A}$ and compute $\alpha_{\mathrm{th}}=\lambda_U/\widehat{\gamma}_U$\;
	
\end{algorithm}

\pagebreak

\section*{Declarations}

\paragraph{Funding}
This research did not receive any specific grant from funding agencies in the public,
commercial, or not-for-profit sectors.

\paragraph{Conflicts of interest}
The author declares that there are no known competing financial interests or personal
relationships that could have appeared to influence the work reported in this paper.

\paragraph{Data availability}
No experimental data were generated or analyzed in this study. Numerical simulations
were performed using synthetic data generated within the model described in the
manuscript.

\paragraph{Code availability}
Code used to generate the numerical results is available from the author upon
reasonable request.

\paragraph{Ethics approval and consent to participate}
Not applicable.

\paragraph{Author contributions}
The author is solely responsible for the conception of the study, the mathematical
analysis, the numerical experiments, and the writing of the manuscript.

\pagebreak

	\bibliographystyle{plain}
	\bibliography{memory_switching_dddas}
	
\end{document}